\documentclass[a4paper,11pt,onecolumn]{elsarticle}
\usepackage{hyperref}	
	\usepackage{subfig}
	\usepackage{tabularx}
	\usepackage{array,multirow,graphicx}
	\usepackage{float}
	\usepackage{color}
	\usepackage{amssymb, extarrows}
	\usepackage{amsmath}
	\usepackage{amsfonts}
	\usepackage{amscd}
	\usepackage{enumerate}
	\usepackage{url}
	\usepackage{cancel}
	\usepackage[normalem]{ulem}
	\usepackage{algorithm, algorithmic}
	\usepackage{upquote}
	\usepackage{multicol}
	\usepackage{booktabs}
	\usepackage{color}
	\usepackage[export]{adjustbox}
	\usepackage{parcolumns}
	
\usepackage{makecell}
		
	\usepackage{mwe}
	\def\ZZ{\mathbb{Z}}

	\def\cal{\mathcal}
	\def\bf{\mathbf}

	\def\sk_s{\mathsf{sk_s}}

	\def\Pr{\mathrm{Pr}}
	\def\Adv{\mathsf{Adv}}
	\def\OW{\textsf{OW-ID-CPA}}

	\def\u{\bf{u}}

	\def\e{\bf{e}}

	\def\x{\bf{x}}

	\def\L{\Lambda}
	\def\Lp{\Lambda^{\perp}}
	
	\def\s{\bf{s}}

	\def\DuA{\cal{D}_{\L_q^{\u}(A),s}}

	\usepackage{color}
	\usepackage{listings} 
	\usepackage{xspace}  % to allow for text macros that don't eat space 
	\usepackage[misc,geometry]{ifsym}
	
	\newtheorem{theorem}{Theorem}
	\newtheorem{lemma}[theorem]{Lemma}
	\newdefinition{remark}{Remark}
	\newdefinition{definition}{Definition}
	
	\newproof{proof}{Proof}

\begin{document}
	
	\title{Lattice-based Signcryption with Equality Test  in Standard Model}
	
	\author[1]{Huy Quoc Le\corref{cor1}}
	\ead{qhl576@uowmail.edu.au}
	\author[1]{Dung Hoang Duong}
	\ead{hduong@uow.edu.au}
	\author[1]{Partha Sarathi Roy}
	\ead{partha@uow.edu.au}
	\author[1]{Willy Susilo}
	\ead{wsusilo@uow.edu.au}
	\address[1]{ Institute of Cybersecurity and Cryptology, School of Computing and Information Technology, University of Wollongong,
		Northfields Avenue, Wollongong NSW 2522, Australia}

	\author[2]{Kazuhide Fukushima}
	\ead{ka-fukushima@kddi-research.jp}
	\author[2]{Shinsaku Kiyomoto}
	\ead{kiyomoto@kddi-research.jp}
	\address[2]{Information Security Laboratory, KDDI Research, Inc.,\\
		2-1-15 Ohara, Fujimino-shi, Saitama, 356-8502, Japan}

	\cortext[cor1]{Corresponding author}
	
	\begin{abstract}
	A signcryption, which is an integration of a public key encryption and a digital signature, can provide confidentiality and authenticity simultaneously. Additionally, a signcryption associated with equality test allows a third party (e.g., a cloud server) to check whether or not two ciphertexts are encrypted from the same message without knowing the message. This application plays an important role especially in computing on encrypted data. In this paper, we propose the first lattice-based signcryption scheme equipped with a solution to testing the message equality in the standard model. The proposed signcryption scheme is proven to be secure against insider attacks under the learning with errors  assumption and the intractability of the short integer solution  problem.  As a by-product,  we also show that some existing lattice-based signcryptions  either is insecure or  does not work correctly. 
	\end{abstract}
	
	\begin{keyword}
	Signcryption, equality test,  standard model, learning with errors problem, short integer solution problem, insider attacks
	\end{keyword}

	\maketitle
	
%	\linenumbers
	\section{Introduction} \label{intro} 
	 
	A signcryption scheme, first proposed by~Zheng \cite{Zhe97}, simulatneously plays the roles of public key encryption and digital signature. A signcryption scheme therefore guarantees the confidentiality and the authenticity at the same time. On the other hand, signcryptions are designed aiming to be more efficient than the signing-then-encrypting approach in terms of cost. 
	
	With a rapid increasing of amount of data, there are more and more personal users as well as organizations transferring their data to third-party service providers for outsourcing. In order to protect their data's sensitive information, data are usually encrypted. This fact requires service providers  to have efficient methods of encrypted data management.  
	Among such tools, equality test (ET) allows service providers to check whether or not  different ciphertexts are generated on the same message even though service providers do not know what the message actually is. This augmented property has been realized in various public key encryption schemes (which are called PKEET), e.g., \cite{YTHW10}, \cite{LLS+20}, \cite{LSQ18}, \cite{Duongetal2019}, \cite{DFS+20} with applications to internet-based personal health record systems \cite{Tan11}, secure outsourced database managements \cite{YTHW10} just to name a few. Of course, the equality test mechanism can also be realized  in signcryptions which we call signcryption with equality test (SCET). SCET has also found some applications such as in  securing messaging services \cite{WPD+19},   industrial Internet of Things \cite{XZH+20}.
	
	Quantum computers are proven to be able to successfully break number-theoretic assumptions, such as the integer factorization problem or the discrete logarithm problem, which are currently the underlying hard problems for a plenty of cryptographic primitives \cite{Sho97}. Under the threat, research community has recently been paying more and more attention as well as resources to the so-called lattice-based cryptography, which based on hard lattice problems. With some advantages of easy implementation, provable hardness, lattice problems (e.g., learning with errors problem (LWE), short integer solution problem (SIS) are playing the role of underlying hard problems for numerous cryptographic primitives.  \\

	\noindent \textbf{Related works.}  %Public key encryption with equality test (PKEET) was first proposed by Yang et al. \cite{YTHW10} implemented in a bilinear group that can be used in different applications on encrypted data such as searchable encryption and partitioning.	Later, Lee et al. \cite{LLS+20} gave a generic construction for PKEET directly leading to the first standard model and the first lattice-based construction.  In \cite{LSQ18}, Lin et al.  gave a generic approach for PKEET constructions which can be easily extended to the identity-based setting.   Duong et al. \cite{Duongetal2019} proposed a direct lattice-based PKEET in the standard model (SDM). Recently, Duong et al. \cite{DFS+20} presenetd a PKEET supporting flexible authorization in the standard model. 
	Malone-Lee and Mao \cite{LM03} in 2003 presented a RSA-based signcryption scheme in the random oracle model. In the same year, Boyen \cite{Boy03} proposed a stringent security model for the schemes which he call \textit{a joint identity-based signature/encryption} (IBSE), and presented an  efficient IBSE construction, based on bilinear pairings.   The work \cite{LSQ18} by Lin et al. gave a generic SCET construction, extended from a generic PKEET construction.  Wang et al. \cite{WPD+19} presented a concrete SCET construction, which the authors call public key signcryption scheme with designated equality test on ciphertexts. However, the primitive in \cite{WPD+19} is based on bilinear groups. More recently, Xiong et al. \cite{XZH+20} propose the so-called  heterogeneous SCET based on pairings which is claimed to be suitable for the sophisticated heterogeneous network of industrial Internet of Things.  \textit{So far, there has been no any SCET construction in the lattice setting}.
	
	Regarding lattice-based signcryption constructions, there have been some works such as \cite{Lietal12},  \cite{GM18},  \cite{YCX19}, \cite{LWJ+14} and  \cite{SS18}.  Li et al. \cite{Lietal12} built a lattice-based signcryption scheme that are only secure in the random oracle model (ROM). Gérard  and Merckx  \cite{GM18}  proposed a lattice-based signcryption scheme in the random oracle model (ROM) based  on the ring learning with errors (RLWE) and a special version of the short integer solution (SIS) offering  the indistinguishability under chosen plaintext attacks (IND-CPA) and  existential unforgeability under chosen message attacks (EUF-CMA) security.  Yang et al. \cite{YCX19} proposed an efficient lattice-based signcryption scheme in the standard model based on RLWE and the ideal short integer solution (ideal-SIS) assumption. The  signcryption  of \cite{YCX19} offers the IND-CCA2 security and  and existential unforgeability under an adaptive chosen-message attack (EUF-aCMA).  Lu et al. \cite{LWJ+14} proposed a lattice-based signcryption scheme without random oracles which is claimed to  achieve  the IND-CCA2 security and EUF-CMA security basing on the hardness of  LWE and SIS. Sato and Shikata \cite{SS18} constructed a lattice-based signcryption in the standard model of which security also based on LWE and SIS. 
We claim that, the IND-CCA2 security of the construction by Lu et al. \cite{LWJ+14} can be easily broken even by a CPA attack (see Section \ref{attack} below) while the  work \cite{SS18} has some serious flaws described as follows. \\

	\noindent \textbf{Description and flaws of the signcryption proposed by  \cite{SS18}.}  The work  \cite{SS18} exploits the  gadget-based trapdoor mechanism proposed by \cite{MP12}, namely, the algorithms \textsf{GenTrap}, \textsf{Invert}, \textsf{SampleD} (see Definition \ref{gtrapdoor} and Lemma \ref{trapdoor} below for further details). Then, for each receiver we generate the public key 
		 $pk_r=\mathbf{A}_r$, and  the private key $sk_r=\textbf{T}_r$ which is a $\mathbf{G}$-trapdoor for $\mathbf{A}_r$  with tag $\mathbf{H}=\textbf{0}$ , i.e., $ \mathbf{A}_r\bigl[\begin{smallmatrix}
		 	 			\mathbf{T}_r\\ \textbf{I}
		 	 			\end{smallmatrix} \bigr] = \mathbf{0} \pmod q $, where $\textbf{G}$ is the gadget matrix specified in \cite[Section 4]{MP12}. The same thing $pk_s=\mathbf{A}_s, sk_r=\textbf{T}_s$ is done for each sender except that $\mathbf{H}=\textbf{I}$, i.e., $ \mathbf{A}_s\bigl[\begin{smallmatrix}
		 			\mathbf{T}_s\\ \textbf{I}
		 			\end{smallmatrix} \bigr] = \mathbf{G} \pmod q $.

	For signcrypting a plaintext  $\mu$, one  utilizes the Dual-Regev framework based on the LWE problem, i.e., one computes $(\overline{\textbf{c}}_0)^t:=\textbf{s}^t\textbf{A}_{r,\textbf{t}}+(\textbf{x}_0)^t$ and $(\overline{\textbf{c}}_1)^t:=\textbf{s}^t\textbf{U}+(\textbf{x}_1)^t$. Here $\textbf{U}$ is a uniform matrix while the matrix $\textbf{A}_{r,\textbf{t}}$  depends  on $\textbf{A}_{r}$ and some vector tag $\textbf{t}$ which relates to the public key of the receiver $pk_r$, to the public key of the sender and to a random small vector $\textbf{r}_e$. Then one signs on the tuple $(\mu|pk_r|\overline{ct})$, where $\overline{ct}=(\overline{\textbf{c}}_0, \overline{\textbf{c}}_1, \textbf{r}_e)$, to get the signature $(\textbf{e}, \textbf{r}_s)$. To do that, one calculates a vector tag $\textbf{h}$ using $pk_s$, $pk_r$, $(\mu|pk_r|\overline{ct})$ and a random small vector $\textbf{r}_s$ and then signs on $\textbf{h}$ using the same way as the signature scheme in \cite[Section 6.2]{MP12}. Finally, the signcryption finishes by computing  $\textbf{c}_0:=\overline{\textbf{c}}_0+\textbf{r}_s$,  $\textbf{c}_1:=\overline{\textbf{c}}_1+\mu\cdot \lfloor q/2 \rfloor$ and outputting the ciphertext $ct=(\textbf{c}_0, \textbf{c}_1, \textbf{r}_e, \textbf{e})$.  
	
	In the unsigncryption algorithm, one can use \textsf{SampleD} to obtain a matrix $\textbf{E}$ such that $\textbf{A}_{r,\textbf{t}}\textbf{E}=\textbf{U} \bmod q$ which helps us to recover $\mu$. However, in this algorithm, one also needs to verify whether the ciphertext is valid or not which needs to recover $\overline{ct}=(\overline{\textbf{c}}_0, \overline{\textbf{c}}_1, \textbf{r}_e)$ from $ct$ first. The way of \cite{SS18} to do that is to run \textsf{Invert} on input $\textbf{c}_0$ to get $\textbf{r}_s$  and then compute $\overline{\textbf{c}}_0=\textbf{c}_0-\textbf{r}_s$, $\overline{\textbf{c}}_1=\textbf{c}_1-\mu \cdot \lfloor q/2 \rfloor $. Unfortunately, we can see that this is not correct since \textsf{Invert} will actually output the sum $\textbf{x}_0+\textbf{r}_s$ instead  of $\textbf{r}_s$. 
	
	One more flaw is that the dimensions of some matrices and some vectors in  the signcryption of \cite{SS18} do not match. Furthermore, the security proofs are quite unclear. For example, in the security proof for the strong unforgeability against insider attacks (i.e., MU-sUF-iCMA) in \cite[Theorem 2]{SS18}, after showing that $\textbf{A}_S \cdot \textbf{z}=0 \bmod q$, the authors do not prove why $\textbf{z}\neq \textbf{0}$.	\\

	\noindent \textbf{Our contribution and technical overview.} In this paper, we propose, for the first time, a lattice-based signcryption scheme possessing a capacity of equality test provably secure in the standard model.  Moreover, since both \cite{LWJ+14} and \cite{SS18} do not work correctly then our work without the equality test part can be considered as a  lattice-based signcryption alternative to them.  For our construction, we consider the multi-user setting and  the insider security model in which there are multiple users and some of them could adversarially behave. We call such users the \textit{internal users} or \textit{insider attackers}. Such kind of attacker is stronger than external adversaries since they can know private information of other users in the setting. We show that our proposed scheme offers OW-iCCA1, IND-iCCA1 and UF-iCMA security against insider attacks relying on the hardness of decisional-LWE and SIS problems. 
	
 Our scheme is basically inspired from the work of Sato and Shikata \cite{SS18} and the recent method of~Duong et al.~\cite{Duongetal2019} for equality test.  We have shown above that \cite{SS18} does not work correctly.  Fortunately, we successfully fix this error simply by in the signcryption algorithm setting $\textbf{c}_0:=\overline{\textbf{c}}_0$ instead of $\textbf{c}_0:=\overline{\textbf{c}}_0+\textbf{r}_s$ and outputting $ct=(\textbf{c}_0, \textbf{c}_1, \textbf{r}_e, \textbf{r}_s, \textbf{e})$ instead of $ct=(\textbf{c}_0, \textbf{c}_1, \textbf{r}_e,  \textbf{e})$. Also, to fix the dimension-related flaw in \cite{SS18}, we use the hash functions named $H_1, H_3$ to make dimensions match.  
	
	% We also exploit the gadget-based trapdoor mechanism \cite{MP12}. It is also worth noting that the work \cite{SS18} has some flaws in its design which is fixed in our work. Roughly speaking, with a special gadget matrix $\textbf{G}$ (publicly known), we say that a matrix $\textbf{T}$ to be $\textbf{G}$-trapdoor with tag $ \mathbf{H} $ for $\textbf{A}$ if $ \mathbf{A}\bigl[\begin{smallmatrix}\mathbf{T}\\ \textbf{I}\end{smallmatrix} \bigr] = \mathbf{H}\mathbf{G} \pmod q $ for some $ \mathbf{H} $ not neccesarily invertible.   

	For equality test, we use an one-way  hash function $H$ and encrypt $H(\mu)$ (but do not sign) in the same way described above for the plaintext $\mu$. This releases some things named $\textbf{c}'_0,$ $ \overline{\textbf{c}}'_1, \textbf{r}'_e, $ corresponding to $\textbf{c}_0,$ $ \overline{\textbf{c}}_1, \textbf{r}_e, $ for $\mu$. Note that, the signing phase (which produces the signature $(\textbf{e}, \textbf{r}_s)$) now runs on input $(\mu|pk_r|\overline{ct})$ with $\overline{ct}=(\textbf{c}_0, \overline{\textbf{c}}_1, \textbf{r}_e, \textbf{c}'_0, \overline{\textbf{c}}'_1, \textbf{r}'_e)$ instead of $\overline{ct}=(\textbf{c}_0, \overline{\textbf{c}}_1, \textbf{r}_e)$. Therefore, the final ciphertext for the proposed SCET is   $ct=(\textbf{c}_0, \textbf{c}_1, \textbf{r}_e, \textbf{r}_s, \textbf{c}'_0,$ $ \overline{\textbf{c}}'_1, \textbf{r}'_e, \textbf{e})$. Finally,  two ciphertexts are proven to come from the same plaintext  if we can recover the same hash value $H(\mu)$ from them without knowing the plaintext $\mu$. 
	
	Furthermore, for the security proof, we also utilize the so-called abort-resistant hash functions defined by \cite[Section 7]{ABB10} provided in Lemma \ref{lemma3}.  Also note that, the presence of $\textbf{B}$, $\textbf{r}_e, \textbf{r}'_e, \textbf{r}_s$ helps us to simulate the responses to the adversary's queries.  \\ 
	
\begin{table}[h]
  	\centering

  	  	\caption{ Some SC and SCET constructions based  on hard lattice problems in the literature. }
  	\medskip
  	\smallskip
  	%	\raisebox{\dimexpr 0.6\baselineskip-\height}% align tops
  	\small\addtolength{\tabcolsep}{0pt}
  	\begin{tabular}{ c | c| c|c|c|c}
  	\hline
   		\textbf{Works}&\textbf{Assumption}& \makecell{\textbf{Security }\\ \textbf{Level}}&\makecell{\textbf{Security}\\ \textbf{Model}}&\textbf{with ET}& \makecell{ \textbf{Insider}\\\textbf{Attacks}}\\
\hline
\hline
\makecell{Li \cite{Lietal12} }&LWE \& SIS&\makecell{IND-CCA2\\ SUF-CMA}&ROM&$\times$ &$\times$\\
\hline
\makecell{Lu \cite{LWJ+14}}& \multicolumn{5}{c}{ Not secure even with IND-CPA (see Section \ref{attack})} \\% LWE \& SIS&\makecell{Not secure}&SDM&$\times$ &$\times$\\
\hline
\makecell{Sato \cite{SS18}}& \multicolumn{5}{c}{Does not work correctly} \\%LWE \& SIS&\makecell{\\\\}&ROM&$\times$ &$\times$\\

\hline 
\makecell{Gérard \cite{GM18} }&RLWE \& SIS&\makecell{IND-CPA\\ EUF-CMA}&ROM&$\times$ &$\times$\\

\hline 
\makecell{Yang \cite{YCX19} }&\makecell{RLWE \\\& ideal-SIS}&\makecell{IND-CCA2\\ EUF-aCMA}&SDM&$\times$ &$\times$\\

\hline
\hline
\textbf{This work}&LWE \& SIS&\makecell{IND-iCCA1\\OW-iCCA1\\SUF-iCMA}&SDM&$\checkmark$ &$\checkmark$\\
	 	\hline
  	\end{tabular} 
  
  	\label{tab2}
\end{table}

	\iffalse
		\begin{table}[pt]

		\centering
		\medskip
		\smallskip
		%	\raisebox{\dimexpr 0.6\baselineskip-\height}% align tops
		\small\addtolength{\tabcolsep}{-3pt}
		\begin{tabular}{  c | c | c  }
			\hline
			\textbf{Ciphertext} $ct$&\textbf{Public key} $pk_r, pk_s$&\textbf{Secret key} $sk_r, sk_s$\\
			\hline\hline
			$3m\cdot  D_{\alpha q}+(m+nk)\cdot D_{\sigma_2}+2(m+\ell)\cdot \mathbb{Z}_q$ &$2(N+M)mn \cdot  \mathbb{Z}_q$ &$2(N+M)\overline{m}nk\cdot D_{\sigma_1}$ \\
		
			\hline
		\end{tabular} 
			\caption{Sizes of our \textsf{SCET} with other constructions. Data sizes are in number of field elements.. Here $a \cdot S$ means $a$ elements in the domain $S$, for example $3mD_{\alpha q}$ indicates $3m$ elements such that each sampled from $D_{\alpha q}$.  }
		\label{tab3}
	\end{table}

\fi

\begin{table}[pt]

	\centering
	\medskip
	\smallskip
	%	\raisebox{\dimexpr 0.6\baselineskip-\height}% align tops
	\small\addtolength{\tabcolsep}{-3pt}
	\begin{tabular}{  |c |   c | }

		\hline
		\textbf{Category}  & \textbf{Size }  \\
		\hline
		Public key per user & $2mn \cdot  \mathbb{Z}_q$   \\
		\hline
			Secret key per user  &$2\overline{m}nk \cdot D_{\sigma_1}$ \\
		\hline
	Ciphertext &$3m\cdot D_{\alpha q}+(m+nk)\cdot D_{\sigma_2}+2(m+\ell)\cdot \mathbb{Z}_q$ \\
	\hline
	
	\end{tabular} 
	\caption{Sizes of our \textsf{SCET}. Here $a \cdot S$ means $a$ elements in the domain $S$. For example $3m\cdot D_{\alpha q}$ indicates that there are $3m$ elements each of which sampled from $D_{\alpha q}$.  }
	\label{tab3}
\end{table}

	\noindent \textbf{Organisation.} In Section \ref{pre}, we give a background of lattices. The framework of signcryption with equality test will be provided in Section \ref{scet}. Section \ref{lbscet} is for our lattice-based signcryption construction. The security of the proposed scheme will be given in Section \ref{security}. Parameter setting will be done in Section \ref{para}. We demonstrate an attack against the IND-CPA of the signcryption construction proposed by Lu et al.  \cite{LWJ+14}   in Section \ref{attack}. In Section \ref{conc}, we make some conclusions on our work. 
	
	\section{Preliminaries} \label{pre}
%	For a positive integer $n$, we write $[n]$ for the set $\{1,\cdots, n\}$. A column vector is written in lower-case bold letter (e.g., $\mathbf{v}$) and a column matrix is written in upper-case bold letter (e.g., $\mathbf{A}$). For a matrix $\mathbf{A}$, notation $\mathbf{A}^t$ presents the $\mathbf{A}$'s transpose. The Gram-Schmidt orthogonalization of a  matrix $\mathbf{A}$ is denoted as $\tilde{\mathbf{A}}$.
 
Throughout this work, the norm $\Vert \mathbf{S} \Vert$ of a set of vectors
	$\mathbf{S}=\{\mathbf{s}_1,\cdots, \mathbf{s}_n\}$ is computed as $\max_{i \in[n]} \Vert \mathbf{s}_i \Vert$. \\

	\noindent \textbf{Lattices.} Given a matrix $\mathbf{B}=[\mathbf{b}_1, \cdots, \mathbf{b}_m]\in \mathbb{R}^{n \times m}$ of $m$ linearly independent vectors,  the set
	$\mathcal{L}(\mathbf{B}):=\{\sum_{i \in [m]}\mathbf{b}_iz_i: z_i\in \mathbb{Z}\}$ is called a lattice of basis $\mathbf{B}$. In this work, we focus on the so-called $q$-ary lattices:
	$ \Lambda_q^{\bot}(\mathbf{A})=\{\mathbf{e} \in \mathbb{Z}^m: \mathbf{A}\mathbf{e}=\mathbf{0} \text{ (mod } q) \},$, and $\L_q^{\bf{u}}(\textbf{A}) :=  \left\{ \e\in\ZZ^m~\rm{s.t.}~A\e=\mathbf{u} \text{ (mod } q) \right\}$, where  $\mathbf{A} \xleftarrow{\$} \mathbb{Z}^{n \times m}$ is a random matrix. Note that, if $\bf{t}\in\L_q^{\bf{u}}(\textbf{A})$ then $\L_q^{\bf{u}}(\textbf{A})=\bf{t}+\Lp_q(\textbf{A})$.
	
	The \textit{first minimum} of a lattice $\mathcal{L}$ is defined as      $\lambda_1(\mathcal{L}):=\min_{\mathbf{v} \in \mathcal{L} \setminus \{\textbf{0}\}}\Vert \mathbf{v}\Vert$. The \textit{$i$-th minimum} of a lattice $\mathcal{L}$ of dimension $n$ is denoted by and defined as      $\lambda_i(\mathcal{L}):=\min\{r: \dim(\text{span}(\mathcal{L} \cap \mathcal{B}_n(0,r))) \geq i\}$, where $\mathcal{B}_n(0,r)=\{\mathbf{x} \in \mathbb{R}^n: \Vert \mathbf{x}\Vert \leq r \}$.    The \textsf{SIVP}$_\gamma$ and \textsf{GapSVP}$_\gamma$  are assumed to be the worst case hard problems in lattices. Given $\mathbf{A}$ to be a basis of a lattice $ \mathcal{L}(\mathbf{A})$ and a positive real number $d$, the first problem requires to find a set of $n$ linearly independent lattice vectors $\mathbf{S} \subset \mathcal{L}(\mathbf{A})$ such that $\Vert  \mathbf{S} \Vert \leq \gamma \lambda_n(\mathbf{A})$, while
	the second one asks to decide if $\lambda_1(\mathcal{L}(\mathbf{A}))\leq d$ or $\lambda_1(\mathcal{L}(\mathbf{A}))>\gamma d$.\\

	\noindent  \textbf{Gaussians.}  Let $m\geq 1$, a vector $\mathbf{c}\in \mathbb{R}^m$ and a positive parameter $s$, for $\mathbf{x}\in \mathbb{R}^m$ define $\rho_{s,\mathbf{c}}(\mathbf{x})= \exp({{-\pi \Vert \mathbf{x}-\mathbf{c}\Vert^2 }/{ s^2}})$. 
	The continuous Gaussian distribution on $\mathbb{R}^m$ with mean $\textbf{c}$ and with width parameter $s$ is proportional to $\rho_{s,\mathbf{c}}(\mathbf{x})$.  
	Let $\mathbb{T}=\mathbb{R}/\mathbb{Z}$ be the additive group of real numbers modulo 1. Given $\alpha>0$ and $m=1$,   we denote by $\Psi_{\alpha}$ the continuous Gaussian distribution on  $\mathbb{T}$ of mean $0$ and width parameter $s:=\alpha$. Remind that, this Gaussian distribution has standard deviation of $\sigma=\alpha/\sqrt{2\pi}$.

	\begin{definition}[Discretized Gaussian]  \label{def11}
	 The discretized Gaussian distribution $\widetilde{\Psi}_{\alpha q}$ is defined by sampling $X \leftarrow \Psi_{\alpha}$ then outputting $\lfloor q \cdot X\rceil \text{ mod } q$. 
	\end{definition}
	In particular, we can define a Gaussian distribution over a subset of $\mathbb{Z}^m$, hence over any integer lattices in  $\mathbb{Z}^m$. 
	\begin{definition}[Discrete Gaussian] \label{def12}
		Let $m$ be a positive integer,  $\Lambda \subset \mathbb{Z}^m$ be any subset,  a vector $\mathbf{c}\in \mathbb{R}^m$ and a positive parameter $s$, define   $\rho_{s, \mathbf{c}}(\Lambda ):=\sum_{\mathbf{x} \in \Lambda } \rho_{s, \mathbf{c}}(\mathbf{x})$. The discrete Gaussian distribution over $\Lambda$ centered at $\mathbf{c} \in \mathbb{Z}^m$ with width parameter $s$ is defined by:
		$\forall \mathbf{x}\in \Lambda$, $D_{\Lambda,  s, \mathbf{c}}(\mathbf{x}):=\rho_{s,\mathbf{c}}(\mathbf{x})/\rho_{s,\mathbf{c}}(\Lambda).$ If $\mathbf{c}=\mathbf{0}$, we just simply write  $\rho_{s}$, $D_{\Lambda,s}.$ If $\Lambda=\mathbb{Z}$,  we can write $D_{\Lambda,s}$ as $D_{s}.$
	\end{definition}
	Note that in Definition \ref{def12}, $\Lambda$ can be a lattice over $\mathbb{Z}^m$. 	The following lemma shows the min-entropy of a discrete Gaussian.
	
		\iffalse A continuous Gaussian with width parameter $s:=\alpha q$ has standard deviation $\sigma=\frac{\alpha q} {\sqrt{2 \pi}}$, hence the discrete Gaussian with width parameter $\alpha q$. It is known that this still roughly holds for the discretized Gaussian, as long as
	$\sigma$ is greater than the smoothing parameter  (see Definition \ref{smt} below) $\eta_{\epsilon}(\mathbb{Z})$ of $\mathbb{Z}$. (See for the discussion in \cite[page 170]{APS15}). This means that we  can identify the discretized $\widetilde{\Psi}_{\alpha q}$ with the discrete $D_{\alpha q}$ if $\frac{\alpha q} {\sqrt{2 \pi}} \geq \eta_{\epsilon}(\mathbb{Z})$.
	
		We present the definition of smoothing parameters below. Note that, smoothing parameters is also utilized to measure the quality of a lattice. 
	
	\begin{definition}[Smoothing Parameters, \cite{MR07}]\label{smt}
		 For any $n$-dimensional lattice $\Lambda$ and positive real $\epsilon>0$, the smoothing parameter $\eta_{\epsilon}(\Lambda)$ is the smallest real $s>0$ such that $\rho_{1/s}(\Lambda^*\setminus \{\mathbf{0}\}) \leq \epsilon$, where $\Lambda^*$ is the dual lattice of $\Lambda$, i.e.,  $\Lambda^*:=\{\textbf{x}\in \text{span}(\Lambda): \langle \textbf{x}, \textbf{y} \rangle \in \mathbb{Z}, \forall \textbf{y}\in \Lambda\}$.
	\end{definition}
	\fi 

		\begin{lemma} {{\cite[Lemma 2.1]{PR06}}} \label{min-entropy}
	Let $\Lambda \subset \mathbb{R}^n$	be a lattice and $s \geq 2\eta_{\epsilon}(\Lambda)$ for some $\epsilon\in (0,1)$. Then for any $\mathbf{c}\in \mathbb{R}^n$ and any $\mathbf{y}\in \Lambda +\mathbf{c}$, $\Pr[\mathbf{x} \gets D_{\Lambda +\mathbf{c},s}: \mathbf{x}=\mathbf{y}]\leq 2^{-n}\cdot \frac{1+\epsilon}{1-\epsilon}$.
	
		\end{lemma}

		\begin{lemma} {{\cite[Lemma 4.4]{MR07}}} \label{bound}
			Let $q> 2$ and let $\mathbf{A}$ be a matrix in $\ZZ_q^{n\times m}$ with $m>n$. Let $\mathbf{T}_A$ be a basis for $\Lp_q(\mathbf{A})$. Then, for $s\geq\|\widetilde{\mathbf{T}_A}\|\cdot  \omega(\sqrt{\log n})$,  $$\Pr[\x\gets\DuA~:~\|\x\|>s\sqrt{m}]\leq\mathsf{negl}(n).$$
	
		\end{lemma}

	\noindent  \textbf{Lerning with Errors problem (LWE) and Short interger Solutions problem (SIS).} Let $n$ and $q \geq 2$ be positive integers and $\chi$ be a distribution on $\mathbb{Z}_q$. Given a vector $\mathbf{s} \xleftarrow{\$} \mathbb{Z}_q^n$, we define an LWE distribution $\mathcal{L}_{\mathbf{s}, \chi}$ on $\mathbb{Z}_q^n \times \mathbb{Z}_q$ as follows: first sample uniformly at random  $\mathbf{a} \xleftarrow{\$} \mathbb{Z}_q^n$, then draw according to $\chi$ an error term $e$, and finally output the pair $(\mathbf{a}, b=\langle \mathbf{a} ,\mathbf{s} \rangle+e \text{ (mod } q))$.

	\begin{definition}[{LWE, \cite{Reg05}}]\label{dlwe}
		The decisional-LWE problem ($\mathsf{dLWE}_{n,q,\chi}$) asks to distinguish a pair $(\mathbf{a},b)\leftarrow \mathcal{L}_{\mathbf{s}, \chi} $ from a pair  $(\mathbf{a},b)\xleftarrow{\$}\mathbb{Z}_q^n \times \mathbb{Z}_q$. %The search-LWE Problem ($\mathsf{sLWE}_{n,q,\chi}$) requires to find the secret vector  $\mathbf{s}$ given polynomially many pairs $(\mathbf{a},b=\langle \mathbf{a} ,\mathbf{s} \rangle+e \text{ (mod } q))\leftarrow \mathcal{L}_{\mathbf{s}, \chi} $.
	\end{definition}

	In the case that $\chi=\widetilde{\Psi}_{\alpha q}$, we instead use the notations $ \textsf{dLWE}_{n,q,\alpha}$ and $ \textsf{sLWE}_{n,q,\alpha}$, 
	 %It is well-known that given a tuple of parameter $(n,q,\alpha)$, the two above \textsf{LWE} problems are equivalent in the sense that if we are able to decide \textsf{DLWE}$_{n,q,\alpha}$, then we can also solve \textsf{SLWE}$_{n,q,\alpha}$ and vice versa; see e.g.,  \cite{Reg05, Pei09, MP12}. Hence, 
	 and generally mention them as the \textsf{LWE}$_{n,q,\alpha}$.  
	Regarding the hardness of \textsf{LWE}, we have the following result: %Regev \cite{Reg05} provided a \textit{quantum reduction} from worst-case lattice problem \textsf{GapSVP} (also,  \textsf{SIVP}) to \textsf{LWE}. Peikert \cite{Pei09} thereafter showed the \textsf{LWE} problem is \textit{classically} at least as hard as those problems if $q$ is sufficiently large.
	\begin{theorem}[{\cite[Theorem 2.16]{BLP+13}}] \label{theo1}
		 Let $n, q\geq 1$ be integers and let $\alpha \in (0,1)$ be such that $\alpha q \geq 2\sqrt{n}$. Then there exists a quantum reduction from worst-case $\mathsf{GapSVP}_{\widetilde{O}(n/\alpha)}$ to $\mathsf{LWE}_{n,q,\alpha}$. If in addition $q\geq 2^{n/2}$ then there is also a classical reduction between those problems.
		
	\end{theorem}

	\begin{definition}[SIS] \label{def2}
		For an integer $q$, a random matrix $\mathbf{A} \xleftarrow{\$} \mathbb{Z}_q^{n \times m}$, and a positive real number $\beta$, the short integer problem $\mathsf{SIS}_{q,n,m, \beta}$ is to find a non-zero vector $\mathbf{z}\in \mathbb{Z}^{m} \setminus \{\mathbf{0}\}$ satisfying that $\mathbf{A}\mathbf{z} =\mathbf{0} \text{ (mod } q)$ and $\Vert \mathbf{z}\Vert \leq \beta.$
	\end{definition}
	
	%The SIS problem was first introduced in by Ajtai in \cite{Ajt96} where the author proved that solving the SIS problem can be reduced to solving certain worst-case problems in lattices. After that, 
	
	It is shown in \cite{MR07} and then in \cite{GPV08} that for large enough $q$, solving SIS is as hard as solving SIVP problem.  Formally, 
	
	\begin{lemma}[{\cite[Proposition 5.7]{GPV08}}] For any poly-bounded $m$, and $\beta=\textsf{poly}(n)$, and for any prime $q \geq \beta\cdot \omega(\sqrt{n\log n})$ the average-case problem $\mathsf{SIS}_{q,n,m,\beta}$ is as hard as $\mathsf{SIVP}_\gamma$ in the worst-case to within certain $\gamma=\widetilde{O}(\beta\sqrt{n})$ factor.
		
		\end{lemma}
	
	The following lemma gives a condition for which the $\mathsf{SIS}_{n,m,q,\beta}$ problem has a solution. 
	
	\begin{lemma}[{\cite[Lemma 5.2]{MR07}}] \label{sissolution}
		For any $q$, $\textbf{A} \in \mathbb{Z}_q^{n \times m}$, and $\beta \geq \sqrt{m}q^{n/m}$, the $ \mathsf{SIS}_{n,m,q,\beta} $ admits a solution.
	\end{lemma}

\noindent \textbf{Gadget-based Trapdoor.} We will recall the notion of \textbf{G}-trapdoor and its related algorithms.

\begin{definition}[\textbf{G}-trapdoors,{\cite[Definittion 5.2]{MP12}}] \label{gtrapdoor}
	Let $ n, q, m, k $ be positive integers and $ \mathbf{A}\in \mathbb{Z}_q^{n\times m} $, $\mathbf{G} \in \mathbb{Z}_q^{n\times nk}$ be matrices with $m \geq nk$. Let $ \mathbf{H}\in \mathbb{Z}_q^{n\times n} $ be some invertible matrix. The $ \mathbf{G} $-trapdoor for $ \mathbf{A} $ with tag $ \mathbf{H} $ is a matrix $ \mathbf{R} \in \mathbb{Z}^{(m-nk)\times nk} $ such that $ \mathbf{A}\bigl[\begin{smallmatrix}
	\mathbf{R}\\ \textbf{I}_{nk}
	\end{smallmatrix} \bigr] = \mathbf{H}\mathbf{G} \pmod q $.  
\end{definition}
The largest singular value $s_1(\textbf{R})$ is used to measure the quality of a \textbf{G}-trapdoor $ \mathbf{R} $ by its. Note that, $s_1(\textbf{R})$  is essentially small as claimed in the following lemma. 
\begin{lemma}[{\cite[Lemma 2.9]{MP12}}] \label{supnorm}
	Let $ D_\sigma^{n\times m} $ be a discrete Gaussian distribution with parameter $ \sigma $ and $ \mathbf{R}\leftarrow D_\sigma^{n\times m} $.  Then with overwhelming probability $ s_1(\mathbf{R}) \le \sigma \cdot \frac{1}{\sqrt{2\pi}}\cdot (\sqrt{n}+\sqrt{m}) $.
\end{lemma}

%============================
Let $ k = \lceil \log_2 q \rceil $, and $ \mathbf{g}^t = (1,2,4,...,2^{k-1}) \in \mathbf{Z}^{k}_q  $.  We will be working with 
$ \mathbf{G} = \mathbf{I}_n\otimes \mathbf{g}^t \in \mathbb{Z}_q^{n\times nk}, $ where $ \otimes $ denotes the tensor product. 
\iffalse
That is,
\[ \mathbf{G} = \left[\begin{matrix}
\mathbf{g}^t&&&\\
&\mathbf{g}^t&&\\
&& \ddots &\\
&&&\mathbf{g}^t\\
\end{matrix}\right] \in \mathbb{Z}_q^{n\times nk}. \]
\begin{thm}[{\cite[Theorem 4.1]{MP12}}]
	For any integers $ q \ge 2, n \geq 1,  k = \lceil \log _2 q \rceil $, there is a primitive matrix $ \mathbf{G} \in \mathbb{Z}_q^{n\times nk} $ such that
	\begin{itemize}
		\item The lattice $ \Lambda^\bot(\mathbf{G}) $ has a known basis $ \mathbf{S}\in \mathbb{Z}^{nk\times nk} $ with $ \|\widetilde{\mathbf{S}}\| \le \sqrt{5} $ and $ \|\mathbf{S}\| \le \max \{\sqrt{5},\sqrt{k}\} $.
		\item Both $ \mathbf{G} $ and $ \mathbf{S} $ require little storage. In particular, they are sparse (with only $ O(nk) $ nonzero entries) and highly structured.
		%	\item Preimage sampling for $ f_{\mathbf{G}}(\mathbf{x})\pmod q $ with Gaussian parameter $ \sigma \ge \|\widetilde{\mathbf{S}}\|\cdot \omega (\sqrt{\log n}) $ can be performed in quasilinear time.
	\end{itemize}
\end{thm}
\fi
%A short special basis $ \mathbf{S}_k \in \mathbb{Z}^{k\times k} $ for $ \Lambda^\bot(\mathbf{g}^t) $ a short matrix  $ \mathbf{S} := \mathbf{I}_n\otimes \mathbf{S}_k \in \mathbb{Z}^{nk\times nk} $ for  $ \Lambda^\bot(\mathbf{G}) $ with  $ \|\widetilde{\mathbf{S}}\| \le \sqrt{5} $ and $ \|\mathbf{S}\| \le \max \{\sqrt{5},\sqrt{k}\} $ are well-known.
 Further details can be found in \cite{MP12}. We will exploit the  following algorithms for the proposed \textsf{SCET} construction. 

%We can easily compute a vector $ \mathbf{x} $ such that $ \mathbf{Gx} = \mathbf{u} \mod q $ for any $ \mathbf{u} \in \mathbb{Z}_q^n $. By Lemma \ref{lem:ProGauss1}, we know that $ \mathbf{x} $ can be small. 

%====================
%===================

	\begin{lemma} \label{trapdoor} 
	Let $ q\ge 2, \overline{m} \ge 1 $, $ k = \lceil \log_2  q \rceil $, and $m = \overline{m}+nk=O(n\log q)$.
	\begin{enumerate}

	\item $(\mathbf{A},\mathbf{R})\leftarrow \mathsf{GenTrap}(n, \overline{m},q, \sigma) $ \cite[Algorithm 1]{MP12}: On input integer $n, \overline{m}, q, \sigma$,  $\mathsf{GenTrap}$ chooses a uniform matrix $ \overline{\mathbf{A}} \in \mathbb{Z}_q^{n\times \overline{m}} $ and a matrix $ \mathbf{H} \in \mathbb{Z}_q^{n\times n} $, then outputs a random matrix $ \mathbf{A} = \left[\overline{\mathbf{A}}|\mathbf{H}\mathbf{G} - \overline{\mathbf{A}}\mathbf{R} \right]  $ and a $ \mathbf{G} $-trapdoor $ \mathbf{R} \sim D_{\sigma}^{\overline{m}\times nk} $ with tag $ \mathbf{H} $. The condition for Gaussian parameter $\sigma$ is that for any $\epsilon \in (0,1)$, $ \sigma  \geq \eta_{\epsilon}(\mathbb{Z})$, i.e.,  $\sigma \geq  \sqrt{\frac{\ln(2(1+1/\epsilon))}{\pi}}$. Moreover,  there exists $\epsilon=\epsilon(n)$ negligible for which $\sigma \geq \omega(\sqrt{\log n})$. Note also that, $ s_1(\mathbf{R}) \le \sigma \cdot \frac{1}{\sqrt{2\pi}}\cdot (\sqrt{\overline{m}}+\sqrt{nk})$ by Lemma \ref{supnorm}. 
	
	\item $\mathbf{e}\leftarrow \mathsf{SampleD}(\mathbf{A},\mathbf{R},\mathbf{H},\mathbf{u},\sigma) $ \cite[Algorithm 3]{MP12}:  On input a matrix $ \mathbf{A} \in \mathbb{Z}_q^{n\times (\overline{m}+nk)} $ and its $ \mathbf{G} $-trapdoor $ \mathbf{R} \in \mathbb{Z}^{\overline{m}\times nk} $, an invertible matrix $ \mathbf{H} \in \mathbb{Z}_q^{n\times n} $, a uniform vector $\mathbf{u}\xleftarrow{\$} \mathbb{Z}_q^n$ and a Gaussian parameter $ \sigma$,   $\mathsf{SampleD}$ outputs a vector $ \mathbf{e} \in \mathbb{Z}^{m+nk} \sim D_{\Lambda_q^{\mathbf{u}}(\mathbf{A}),\sigma} $. The condition for $ \sigma$ is that $ \sigma \geq \sqrt{7( s_1(\mathbf{R})^2+1)}\cdot \omega(\sqrt{\log n})$ (see \cite[Section 5.4]{MP12}).
	
%	\item $\mathbf{R}' \leftarrow \mathsf{DelTrap}(\mathbf{A},\mathbf{A}_1,\mathbf{R},\mathbf{H},\sigma) $ \cite[Algorithm 4]{MP12}:  On input a matrix $ \mathbf{A} \in \mathbb{Z}_q^{n\times m} $ and its $ \mathbf{G} $-trapdoor $ \mathbf{R}\in \mathbb{Z}^{(m-nk)\times nk} $, a matrix $ \mathbf{A}_1 \in \mathbb{Z}_q^{n\times nk}$, an invertible matrix $ \mathbf{H} \in \mathbb{Z}^{n\times n} $ and a Gaussian parameter $ \sigma $,  $\mathsf{DelTrap}$ outputs a $ \mathbf{G} $-trapdoor $ \mathbf{R}' \in \mathbb{Z}^{m\times nk} $ for matrix $ [\mathbf{A}|\mathbf{A}_1] $ with tag $ \mathbf{H} $. The condition for $ \sigma$ is that $ \sigma \geq \eta_{\epsilon}(\Lambda_q^{\bot}(\mathbf{A}))$, hence we can choose $\sigma \ge \sqrt{5}(s_1(\mathbf{R})+1)\cdot \omega(\sqrt{\log n}) $ (see \cite[Lemma 2.3, Lemma 5.3]{MP12}). 

	\item $(\mathbf{s},\mathbf{e})  \leftarrow \mathsf{Invert}(\mathbf{R},\mathbf{A},\mathbf{b}^t=\mathbf{s}^t\mathbf{A}+\mathbf{e}^t) $ \cite[Algorithm 2]{MP12}: On input a uniform matrix  $\mathbf{A}$ and its $\mathbf{G}$-trapdoor $\mathbf{R},$ and a vector $\mathbf{b}$ such that $\mathbf{b}^t=\mathbf{s}^t\mathbf{A}+\mathbf{e}^t$,  $\mathsf{Invert}$ returns $(\mathbf{s}$ and $\mathbf{e})$. Note that if $\textbf{e} \gets D_{\mathbb{Z}^m,\alpha q}$ and $1/\alpha \geq 2 \sqrt{5 (s_1(\mathbf{R})^2+1)}\cdot \omega(\sqrt{\log n})$ then  $\mathsf{Invert}$ succeeds with overwhelming probability over the choice of $\mathbf{e}$ (see \cite[Theorem 5.4]{MP12}).

	\end{enumerate}
	\end{lemma}

	We adapt Lemma 6 in \cite{DM14} for scalar tags (i.e., $\mathbf{H}=x\cdot \mathbf{I}_n$ for some $x \in \mathbb{Z}_q\setminus \{0\}$) to get the following lemma which will be helpful for the security analysis in Section \ref{security}:
	
		\begin{lemma}[{Adapted from  \cite[Lemma 6]{DM14}}] \label{lemma4}
	For $i=0,\cdots, n$, let $\mathbf{T}^{(i)}$ be $\mathbf{G}$-trapdoor for $[\mathbf{A}|\mathbf{A}^{(i)}] \in \mathbb{Z}_q^{n\times (m-k)}\times \mathbb{Z}_q^{n \times k}$ with tag $\mathbf{H}^{(i)}=x_i \mathbf{I}_n$ for some $x_i \in \mathbb{Z}_q\setminus \{0\}$. Then any linear combination $\mathbf{T}=\sum_{i=1}^{n}h_i \mathbf{T}^{(i)}$ with $h_i \in \mathbb{Z}_q$ is a $\mathbf{G}$-trapdoor for $[\mathbf{A}|\sum_{i=1}^{n}h_i \mathbf{A}^{(i)}]$ with tag $\mathbf{H}=\sum_{i=1}^{n}h_i \mathbf{H}^{(i)}=(\sum_{i=1}^{n}h_i x_i)  \mathbf{I}_n \neq \mathbf{0}$.
			\end{lemma}

%======================================
	\section{Framework of Signcryption Scheme with Equality Test} \label{scet}

	There are two settings for an SCET scheme depending on the number of users joining the scheme  \cite{ADR02}.  While in two-user setting, there are only one receiver and only one sender, the multi-user setting involves with multiple receivers and senders.  %Among them, we call \textit{the target receiver}, say Bob (resp. \textit{the target sender}, say Alice)  the receiver (resp. the sender) that is attacked by some adversary.  Thus, in this setting, the public key of a receiver used in each call of \textsf{SC.SC} along with the private key of the target sender Alice is arbitrary (not necessarily the public  key of Bob). Similarly, the public key of a sender used in each call of \textsf{SC.USC} along with the private key of the target receiver Bob is arbitrary (not necessarily the public  key of  Alice). Furthermore, in
	In  this setting, it is supposed that the attacker knows all public keys of all receivers and of all senders when he accesses the communication channel between the target sender  and the target receiver. See \cite[Subsection 1.3]{Baek07} for more details.	
	
	From now on, we suppose that in a SCET scheme, there are $N$ receivers and $M$ senders. We also use $r$ (resp., $s$) to represent the index of a receiver (resp., a sender).

%=========	
	\subsection{Syntax} \label{syn}
	 A SCET is a tuple of algorithms \textsf{Setup},  \textsf{KGr}, \textsf{KGs}, \textsf{SC} and \textsf{USC},  \textsf{Tag},  and \textsf{Test} which is described as follows: 
\begin{itemize}
\item \textsf{Setup($1^{\lambda}$)} is a probabilistic polynomial time (PPT) algorithm that takes as input a security parameter $\lambda$ to output a set of public parameters $pp$.
\item $ \textsf{KGr}(pp)$ (resp., $\textsf{KGs}(pp)$) is a PPT algorithm that on input the set of public parameters $pp$, outputs a public key $pk_r$ and a private key $sk_r$ for a receiver $\mathcal{R}$ (resp., a public key $pk_s$ and a private key $sk_s$ for a sender $\mathcal{S}$). 	  
\item	 \textsf{SC}$(pk_r,sk_s,\mu)$ is a PPT algorithm takes as input  a public key $pk_r$ of a receiver, a private key $ sk_s$ of a sender  and a message $\mu$ in the message space  $\mathcal{M}$  to output a  ciphertext $ct$. 
\item   \textsf{USC$(sk_r,pk_s,ct)$} is a deterministic polynomial time (DPT) algorithm takes as input the private key $ sk_r$ of a receiver,  a public key $pk_s$ of a sender  and a ciphertext $ct$ to output a message $\mu$ or an invalid $\bot$. 
\item	\textsf{Tag$(sk_r)$} is a DPT algorithm that on  input a private key  $sk_r $ of a receiver to output a tag $tg_r$.
\item	\textsf{Test$(tg_1, ct_1, tg_2, ct_2)$} is a DPT algorithm that takes as input two pairs of tag/ciphertext  $(tg_1, ct_1)$, $(tg_2, ct_2)$ to output $1$ if $ct_1$ and $ct_2$ are generated on the same message or $0$ otherwise.
\end{itemize}

	\iffalse
	\begin{figure}
		\medskip
		\centering
		\smallskip
		%	\raisebox{\dimexpr 0.6\baselineskip-\height}% align tops
		\small\addtolength{\tabcolsep}{10pt}
		\frame{\includegraphics[scale=0.7]{gameIND.png}}
		\caption{\textsf{IND-CPA} Game against Type-1 Adversary.}
		\label{fig1}
	\end{figure}
	
	\fi

%	\begin{remark}
%		In our proposed scheme, two key generation algorithms $\mathsf{KGr}(pp)$ and $\mathsf{KGs}(pp)$ are identical. Such a setting is better than a setting of two different key generation methods (as in, eg., \cite{SS18}) in the sense that receivers and senders can easily  exchange their role in the signcryption protocol.
%	\end{remark}
	\subsection{Correctness} \label{cor}
	Let $\lambda$ be any security parameter. For any $pp \leftarrow \textsf{Setup}(1^{\lambda})$, $(pk_r, sk_r) \leftarrow \textsf{KGr} (pp)$, $(pk_s, sk_s) \leftarrow \textsf{KGs} (pp)$,  $(pk_{r_1}, sk_{r_1}) $ $ \leftarrow \textsf{KGr} (pp)$, $(pk_{s_1}, sk_{s_1}) \leftarrow \textsf{KGs} (pp)$, $(pk_{r_2}, sk_{r_2}) \leftarrow \textsf{KGr} (pp)$, $(pk_{s_2}, sk_{s_2}) \leftarrow \textsf{KGs} (pp)$, $tg_1 \leftarrow \textsf{Tag}(sk_{r_1})$ and $tg_2 \leftarrow \textsf{Tag}(sk_{r_2})$, any ciphertexts $ct_1$ and $ct_2$, and any message $\mu \in \mathcal{M}$, the correctness for an SCET scheme requires all the following to hold:
	\begin{enumerate}
		\item $\Pr[\mu=\textsf{USC}(sk_r, pk_s,\textsf{SC}( pk_r, sk_s,\mu))]=1-\textsf{negl}(\lambda).$ \\This says that given a valid ciphertext on a message, the unsigncryption algorithm succeeds in recovering that message with overwhelming probability.
		\item If
		 $\textsf{USC}(sk_{r_1}, pk_{s_1},ct_1)$ $=\textsf{USC}(sk_{r_2}, pk_{s_2},ct_2) \neq \bot,$
		then $$\Pr[\textsf{Test}(tg_1,ct_1,tg_2,ct_2) =1]=1-\textsf{negl}(\lambda).$$ This says that if two ciphertexts are from the same message then the equality test algorithm returns $1$ with overwhelming probability.
		\item If
		$\textsf{USC}(sk_{r_1}, pk_{s_1},ct_1) \neq \textsf{USC}(sk_{r_2}, pk_{s_2},ct_2),$
		then $$\Pr[\textsf{Test}(tg_1,ct_1,tg_2,ct_2)=1]=\textsf{negl}(\lambda).$$ This says that if two ciphertexts are generated on two different messages  then the equality test algorithm returns $1$ with negligible probability.
	\end{enumerate}

	\subsection{Security} \label{sec}
		We categorize the security for SCET into the outsider and  insider securities. In the outsider security setting, an external adversary cannot know private information of users but  public information (e.g., public system parameters and public keys).
	In contrast,  in the insider security setting, an internal adversary can know some
	private keys of other users hence he is stronger than any external adversaries.

%	We define the   \textit{indistinguishability under chosen ciphertext attacks} (IND-CCA) security model for a type of adversary whose goal is to guess which message between two options that is used  to produce the challenge ciphertext. Also, we define the   \textit{one-wayness under chosen ciphertext attacks} (OW-CCA) security model against the adversary who wants to recover the message corresponding to the challenge ciphertext. Finally, we define the \textit{existential unforgeability against chosen message attacks} (UF-iCMA)  security model in which  the adversary's goal is try to forge at least one valid ciphertext on a new message. To these end, we follow the work  \cite{LSQ18}. 

We also consider three types of adversary against a SCET scheme. Remark that, all of them can be internal adavesaries, i.e., insider attackers.	These types of adversary and their behaviors will be detailed in the following definitions and games.

	  \begin{itemize}
	  	\item 
	\textit{Type 1 adversary} is  supposed to know the target receiver, but does not have the tag of the target receiver and his goal is to guess which message between two options that is used in the signcryption algorithm to produce the challenge ciphertext. As an insider attacker, he can also know the target sender's public and private keys.   The Type 1 adversary corresponds to the IND-iCCA1 game.
	
	\item \textit{Type 2 adversar}y  is supposed to know the target receiver and can  perform equality tests on any ciphertexts and his goal is to recover the message corresponding to the challenge ciphertext. As an insider attacker, he can also know the target sender's public and private keys.    The Type 2  adversary corresponds to the OW-iCCA1 game.
	
	\item \textit{Type 3 adversary} is supposed to know the target sender and his goal is  to try to forge at least one valid ciphertext. As an insider attacker, he can also know the target receiver's public and private keys.  The Type 3 adversary corresponds to the UF-iCMA game.
	
	  \end{itemize}
	  
\iffalse	

	\noindent \textbf{Oracles for Adversaries}. In the security games, the adversaries are allowed to make queries to the following oracles:
		\begin{itemize}
		\item private key query  PKQ$(r)$: $\mathcal{A}_1$ submits an index $r $, $\mathcal{C}$ in turn returns the private key sk$_{r}$ of the receiver $\mathcal{R}_{r}$  to $\mathcal{A}_1$. 
		\item Signcryption query SCQ$(r,s,\mu)$: $\mathcal{A}_1$ sends two indexes $r$ and $s$ together with a message $\mu$ to $\mathcal{C}$. In turn, $\mathcal{C}$ sends the output $ct$ of $\textsf{SC}(pk_{r}$, $sk_{s},\mu)$ back to $\mathcal{A}_1$.
		%	\item \textit{Unsigncryption query}  $ USQ(r,s,ct)$: $\mathcal{A}_1$ chooses two indexes $r$ and $s$ with a ciphertext $ct$ to send to $\mathcal{C}$, $\mathcal{C}$ sends the output $\mu$/$\bot$ of $\textsf{USC}(sk_{r},pk_{s},ct)$ back to $\mathcal{A}_1$.
		\item Tag query TGQ$(r)$: $\mathcal{A}_1$ chooses an index $r $ to send to $\mathcal{C}$. The challenger $\mathcal{C}$ in turn sends the output $tg_{r}$ of $\textsf{Tag}(sk_{r})$ back to $\mathcal{A}_1$.
	\end{itemize}
\fi

%In the following, CCA1 means that the adversary is not allowed to make unsigncryption queries after the \textbf{Challenge} phase being done. 

			\begin{definition}[IND-iCCA1] \label{indcca2}
				An SCET scheme is IND-iCCA1 secure if the advantage of  any PPT adversary  $\mathcal{A}_1$ playing the $\mathsf{INDCCA1}^{\mathcal{A}_1}_{SCET}$ game is negligible: $\mathsf{Adv}^{\text{IND-iCCA1}}_{\mathcal{A}_1}(\lambda):= \vert \Pr[\mathsf{INDCCA1}^{\mathcal{A}_1}_{SCET} \Rightarrow 1] -1/2\vert\leq \mathsf{negl}(\lambda).$
				
			\end{definition}
			
	The $\mathsf{INDCCA1}^{\mathcal{A}_1}_{SCET}$ game is defined as follows:
			
		 \textbf{Setup.} The challenger  $\mathcal{C}$ first runs \textsf{Setup($1^{\lambda}$)} to have the set of  public parameters $pp$. Then $\mathcal{C}$ runs $\textsf{KGr}(pp)$ to get $(pk_{r}, sk_{r})$ for $r\in[N]$, and $\textsf{KGs}(pp)$ to get $(pk_{s}, sk_{s})$ for $s\in[M]$, then sends $(pp, \{pk_{r}\}_{r \in [N]}), \{pk_{s}\}_{s \in [M]})$ to the adversary $\mathcal{A}_1$. Let $r^*\in[N]$ be the index of the target receiver.
		 
	 \textbf{Phase 1.}  $\mathcal{A}_1$ adaptively makes a polynomially bounded number of the following queries:
				\begin{itemize}
					\item Private key query  PKQ$(r)$: If $r=r^*$, the challenger $\mathcal{C}$ rejects the query. Otherwise, $\mathcal{C}$ returns the private key $sk_{r}$ of the receiver $\mathcal{R}_{r}$  to $\mathcal{A}_1$. 
					\item Signcryption query SCQ$(r,s,\mu)$: The challenger $\mathcal{C}$ sends the output $ct$ of $\textsf{SC}(pk_{r}$, $sk_{s},\mu)$ back to $\mathcal{A}_1$.
					\item Unsigncryption query   USQ$(r,s,$ $ct)$: $\mathcal{C}$ sends the output of $\textsf{USC}(sk_{r},pk_{s},$ $ct)$ back to $\mathcal{A}_1$.
					\item Tag query TGQ$(r)$: If $r=r^*$, the challenger $\mathcal{C}$ rejects the query. Otherwise, $\mathcal{C}$ in turn sends the output $tg_{r}$ of $\textsf{Tag}(sk_{r})$ back to $\mathcal{A}_1$.
				\end{itemize}
				
		\textbf{Challenge.} $\mathcal{A}_1$ submits two messages $\mu^*_0, \mu^*_1$ together with the target sender's keys $(pk_{s^*}, sk_{s^*})$. The challenger $\mathcal{C}$ then chooses uniformly at random a bit $b \in \{0,1\}$ and returns the challenge ciphertext $ct^* \leftarrow \mathsf{SC}(pk_{r^*}, sk_{s^*},\mu^*_b)$ to $\mathcal{A}_1$.
		
	 \textbf{Phase 2.} $\mathcal{A}_1$ queries the oracles again as in \textbf{Phase 1} with a restriction that $\mathcal{A}_1$ is not allowed to make the query PKQ$(r^*)$ and all unsigncryption queries USQ$(r,s$, $ct)$.
			
		 \textbf{Output.} $\mathcal{A}_1$  outputs a bit $b' \in \{0,1\}$. He wins the game if $b'=b$.

			\begin{definition}[OW-iCCA1] \label{owcca2}
				The scheme SCET is OW-iCCA1 secure if the advantage of  any PPT adversary  $\mathcal{A}_2$ playing the $\mathsf{OWCCA1}^{\mathcal{A}_2}_{SCET}$  game is negligible: $\mathsf{Adv}^{\text{OW-iCCA1}}_{\mathcal{A}_2}(\lambda):= \Pr[\mathsf{OWCCA1}^{\mathcal{A}_   2}_{SCET}\Rightarrow 1] \leq \mathsf{negl}(\lambda).$
			\end{definition}

			The $\mathsf{OWCCA1}^{\mathcal{A}_2}_{SCET}$ game is defined as follows:
	
				\textbf{Setup.} The challenger  $\mathcal{C}$ first runs \textsf{Setup($1^{\lambda}$)} to have the set of  public parameters $pp$ and then runs $\textsf{KGr}(pp)$ to get $(pk_{r}, sk_{r})$ for $r\in[N]$, and $\textsf{KGs}(pp)$ to get $(pk_{s}, sk_{s})$ for $s\in[M]$, then sends $(pp, \{pk_{r}\}_{r \in [N]}), \{pk_{s}\}_{s \in [M]})$  to the adversary $\mathcal{A}_2$. Let $r^*\in[N]$ be the index of the target receiver.
				
			 \textbf{Phase 1.}  $\mathcal{A}_2$ adaptively makes polynomially bounded number of the following queries:
				\begin{itemize}
					\item Private key query  PKQ$(r)$: If $r=r^*$,  $\mathcal{C}$ rejects the query. Otherwise, $\mathcal{C}$ returns the private key $sk_{r}$ of the receiver $\mathcal{R}_{r}$ . 
					\item Signcryption query SCQ$(r,s,\mu)$: $\mathcal{C}$ returns the output $ct$ of $\textsf{SC}(pk_{r},sk_{s},\mu)$.
					\item Unsigncryption query  USQ$(r,s,$ $ct)$:  $\mathcal{C}$ sends the output  of $\textsf{USC}(sk_{r},pk_{s},$ $ct)$ back to $\mathcal{A}_2$.
					\item Tag query TGQ$(r)$: $\mathcal{C}$ returns the output $tg_{r}$ of $\textsf{Tag}(sk_{r})$ (even when $r=r^*$).
				\end{itemize}
				
			 \textbf{Challenge.}  $\mathcal{A}_2$ submits  the target sender's keys $(pk_{s^*}, sk_{s^*})$, $\mathcal{C}$ chooses a random message $\mu^* \in \mathcal{M}$ and returns the challenge ciphertext $ct^*\leftarrow \textsf{SC}(pk_{r^*}, sk_{s^*},\mu^*)$ to $\mathcal{A}_2$.
			 
			\textbf{Phase 2.}  $\mathcal{A}_2$ queries the oracles again as in \textbf{Phase 1} with a restriction that $\mathcal{A}$ is not allowed to make the query PKQ$(r^*)$ and all unsigncryption queries USQ$(r,s$, $ct)$.
			
			\textbf{Output.} $\mathcal{A}_2$  outputs $\mu'^*$. He wins the game if $\mu'^*=\mu^*$.

				\begin{remark} One should be aware that there is no any reduction from IND-iCCA1 to OW-iCCA1 because the OW-iCCA1 adversary is allowed to know the tag of the target receiver, whilst the IND-iCCA1 is not. 
				\end{remark}	
			
			\begin{definition}[UF-iCMA] \label{ufcma}
				The scheme SCET is UF-iCMA secure if the advantage of any PPT adversary  $\mathcal{A}_3$ playing the $\mathsf{UFCMA}^{\mathcal{A}_3}_{SCET}$ game is negligible: $\mathsf{Adv}^{\text{UF-iCMA}}_{\mathcal{A}_3}(\lambda):= \Pr[\mathsf{UFCMA}^{\mathcal{A}_3}_{SCET}\Rightarrow 1]\leq \mathsf{negl}(\lambda).$
			\end{definition}
			
		The $\mathsf{UFCMA}^{\mathcal{A}_3}_{SCET}$ game is defined as follows:

	\textbf{Setup.} The challenger  $\mathcal{C}$ first runs \textsf{Setup($1^{\lambda}$)} to have the set of  public parameters $pp$ and then runs $\textsf{KGr}(pp)$ to get $(pk_{r}, sk_{r})$ for $r\in[N]$, and $\textsf{KGs}(pp)$ to get $(pk_{s}, sk_{s})$ for $s\in[M]$, then sends $(pp, \{pk_{r}\}_{r \in [N]}), \{pk_{s}\}_{s \in [M]})$ to the adversary $\mathcal{A}_3$. Let $s^*\in[M]$ be the index of the target sender.
	
	 \textbf{Queries.}  $\mathcal{A}_3$ adaptively makes polynomially bounded number of the following queries:
				\begin{itemize}
					\item Private key query PKQ$(s)$:
					If $s=s^*$, $\mathcal{C}$ rejects the query. Otherwise, $\mathcal{C}$ returns the private key $sk_{s}$ of the sender $\mathcal{S}_{s}$  to $\mathcal{A}_3$. 
				
					\item Signcryption query SCQ$(r,s,\mu)$: $\mathcal{C}$ returns the output $ct$ of $\textsf{SC}(pk_{r},sk_{s},\mu)$ back to $\mathcal{A}_3$.
					\item Unsigncryption query USQ$(r,s,ct)$:  $\mathcal{C}$ returns the output of $\textsf{USC}(sk_{r},pk_{s},ct)$.
					\item Tag query TGQ$(r)$: $\mathcal{C}$ returns the output $tg_{r}$ of $\textsf{Tag}(sk_{r})$.
				\end{itemize}

		 \textbf{Forge.} $\mathcal{A}_3$  outputs an index $r^* $ of some receiver  and a ciphertext $ct^*$ on a message $\mu^*$, where $ct^*$  must not be the output of any query $ SCQ(r,s,\mu)$ in the query phase. He wins the game if $\mathsf{USC}(sk_{r^*},pk_{s^*},ct^*) \neq \bot$. 
		 Note that, if $\mu^*$ is not the same as the messages queried previously, the SCET scheme  is called EUF-iCMA (i.e., existential unforgeability). If $\mu^*$ is one of the messages queried previously but $(\mu^*, ct^*) \neq (\mu,ct)$ for all $(\mu,ct)$ that was queried previously, then the SCET scheme  is called SUF-iCMA (i.e., strong unforgeability).

	\section{Our Construction} \label{lbscet}
	
	In this section, we describe a lattice-based signcryption with equality test, named \textsf{SCET}. The proposed \textsf{SCET}  signcryption consists of algorithms \textsf{Setup}, \textsf{KG}, \textsf{SC}, \textsf{USC}, \textsf{Tag} and \textsf{Test}.  We also consider  lattice-based collision-resistant hash functions  indicated  by  a uniform matrix $\textbf{W}$  defined as $f_{\textbf{W}}(\textbf{x}):=\textbf{W}\textbf{x} \bmod q$ (cf. \cite{MR07}).
	
	\begin{description}
	
		\item[ \underline{\textsf{Setup($1^{n}$)}}:] On input a security parameter $n$, perform the following:
				\begin{enumerate}
					\item Set parameters $n, q, \overline{m}$, $\ell$, $N$, $M$, $\alpha$,   $ \sigma_1$, $ \sigma_2$,  $k=\lceil \log q \rceil$, $m=\overline{m}+nk$ as in Section \ref{para}.
				\item Samples randomly and independently matrices 
				$\mathbf{C}_0, \cdots, \mathbf{C}_{n},$ $  \mathbf{C}'_0, \cdots, $ $ \mathbf{C}'_{n} \in \mathbb{Z}_q ^{n \times nk}$, $\mathbf{B}, \mathbf{B}' \in \mathbb{Z}_q ^{n \times m}$, $\mathbf{U},\mathbf{U}' \in \mathbb{Z}_q ^{n \times \ell}$. 
				\item Samples randomly vector $ \mathbf{u} \in \mathbb{Z}_q ^{n}$. 
			\item One-way hash function $H:\{0,1\}^\ell \rightarrow \{0,1\}^\ell $, collision-resistant hash functions $H_1:\mathbb{Z}_q^{n \times m} \rightarrow \mathbb{Z}_q^{\overline{m} }$, $H_2$ is a full-rank differences (FRD) encoding \footnote{See \cite[Section 5]{ABB10} for details on FRD.} and a universal hash function $H_3:\{0,1\}^{*} \rightarrow  \mathbb{Z}_q^{\overline{m}}$.
			\item A  plaintext (message) space $\mathcal{M}= \{0,1\}^\ell  $.
			\item Return  $pp=\{n, q, k, \overline{m}, m,\ell, \alpha, \sigma_1, \sigma_2, N, M, \mathcal{M},$ $  (\bf{C}_i, \bf{C}'_i)_{i=0}^{n}, H, H_1, $ $H_2, H_3\}$ as the set of public parameters.
				\end{enumerate}	
		\item [\underline{\textsf{KG($pp$)}}:] On input the public parameters $pp$, do the following:
				\begin{enumerate}
					\item For each receiver $r \in [N]$, generate  $\overline{\mathbf{A}}_r, \overline{\mathbf{A}}'_r \xleftarrow{\$} \mathbb{Z}_q^{n \times\overline{ m}}$, $ \mathbf{T}_r, \mathbf{T}'_r \leftarrow D_{\mathbb{Z}^{\overline{m} \times n k }, \sigma_1} $ and then set $\mathbf{A}_r=[\overline{\mathbf{A}}_r|-\overline{\mathbf{A}}_r\cdot \mathbf{T}_r]\in \mathbb{Z}_q^{n \times m}$, $\mathbf{A}'_r=[\overline{\mathbf{A}}'_r|-\overline{\mathbf{A}}'_r\cdot \mathbf{T}'_r] \in \mathbb{Z}_q^{n \times m}$
					\item Similarly, for each sender $ s\in [M]$, generate  $\overline{\mathbf{A}}_s, \overline{\mathbf{A}}'_s \xleftarrow{\$} \mathbb{Z}_q^{n \times\overline{ m}}$, $ \mathbf{T}_s, \mathbf{T}'_s \leftarrow D_{\mathbb{Z}^{\overline{m} \times n k }, \sigma_1} $ and then set $\mathbf{A}_s=[\overline{\mathbf{A}}_s|\mathbf{G}-\overline{\mathbf{A}}_s\cdot \mathbf{T}_s] \in \mathbb{Z}_q^{n \times m}$, $\mathbf{A}'_s=[\overline{\mathbf{A}}'_s|\mathbf{G}-\overline{\mathbf{A}}'_s\cdot \mathbf{T}'_s] \in \mathbb{Z}_q^{n \times m}$
		%		\item  Sample randomly matrices $\mathbf{B}_r, \mathbf{B}'_r, \mathbf{B}_s, \mathbf{B}'_s \xleftarrow[]{\$} \mathbb{Z}_q^{n \times m}$
				\item Return
		 $pk_r=(\mathbf{A}_r,\mathbf{A}'_r)$, and  $sk_r=(\textbf{T}_r, \textbf{T}'_r)$ as public key and private key for a receiver $\mathcal{R}$ of index $r$, $pk_s=(\mathbf{A}_s,\mathbf{A}'_s)$, and $sk_s=(\textbf{T}_s, \textbf{T}'_s)$ as public key and private key for a sender $\mathcal{S}$  of index $s$.
				\end{enumerate}
	
	\item [\underline{\textsf{SC}($pk_r,sk_s,\mu$)}:] On input a receiver's public key $pk_r=(\mathbf{A}_r,\mathbf{A}'_r)$, a sender's private key $sk_s=(\mathbf{T}_s, \mathbf{T}'_s)$, a plaintext $ \mu \in \mathcal{M}$, perform the following:

				\begin{enumerate}
					\item  $\textbf{r}_e, \textbf{r}'_e \leftarrow D_{\mathbb{Z}^{m }, \alpha q}$, $\textbf{t}=f_{\overline{\textbf{A}}_r}(H_1(\textbf{A}_s))+f_{\textbf{B}}(\textbf{r}_e)\in \mathbb{Z}_q^{n}$,\\
					$\textbf{t}'=f_{\overline{\textbf{A}}'_r}(H_1(\textbf{A}'_s))+f_{\textbf{B}'}(\textbf{r}'_e)\in \mathbb{Z}_q^{n}$.
					\item $\textbf{A}_{r,\textbf{t}}=\textbf{A}_{r}+[\textbf{0}|H_2(\textbf{t})\mathbf{G}]\in \mathbb{Z}_q^{n \times m}=[\overline{\mathbf{A}}_r|H_2(\textbf{t})\mathbf{G}-\overline{\mathbf{A}}_r\cdot \mathbf{T}_r]\in \mathbb{Z}_q^{n \times m}$,\\
					$\textbf{A}'_{r,\textbf{t}}=\textbf{A}'_{r}+[\textbf{0}|H_2(\textbf{t}')\mathbf{G}]\in \mathbb{Z}_q^{n \times m}=[\overline{\mathbf{A}}'_r|H_2(\textbf{t}')\mathbf{G}-\overline{\mathbf{A}}'_r\cdot \mathbf{T}'_r]\in \mathbb{Z}_q^{n \times m}$.
					
					\item $\textbf{s}, \textbf{s}' \xleftarrow{\$} \mathbb{Z}_q^{n}$, \quad $\textbf{x}_1,  \textbf{x}_1'\leftarrow D_{\mathbb{Z}^{\ell}, \alpha q}$.
					\item $\textbf{c}_0=\textbf{s}^t\textbf{A}_{r,\textbf{t}}+ \textbf{x}^t_0\in \mathbb{Z}_q^{m}$, \quad $\overline{\textbf{c}}_1=\textbf{s}^t\textbf{U}+ \textbf{x}^t_1 \in \mathbb{Z}_q^{\ell}$,\\
					$\textbf{c}'_0=(\textbf{s}')^t\textbf{A}'_{r,\textbf{t}} + (\textbf{x}'_0)^t\in \mathbb{Z}_q^{m}$, \quad $\overline{\textbf{c}}'_1=(\textbf{s}')^t\textbf{U}'+ (\textbf{x}'_1)^t \in \mathbb{Z}_q^{\ell}$.
					\item  Set  $\overline{ct}=(\textbf{c}_0, \overline{\textbf{c}}_1, \textbf{r}_e, \textbf{c}'_0, \overline{\textbf{c}}'_1, \textbf{r}'_e)$.
				
					\item Sign on $\mu| pk_r|\overline{ct}$ to get the signature $(\textbf{e}, \textbf{r}_s )$ as follows:
					\begin{enumerate}
					\item  $ \textbf{r}_s \leftarrow D_{\mathbb{Z}^{m }, \alpha q}$.
					\item $\textbf{h}=(h_1, \cdots, h_n)=f_{\overline{\textbf{A}}_s}(H_3(\mu| pk_r|\overline{ct}))+f_{\textbf{B}}(\textbf{r}_s)\in \mathbb{Z}_q^{n}$.
					\item $\textbf{A}_{s,\textbf{h}}=[\mathbf{A}_s|\mathbf{C}_0+\sum_{i=1}^{n}h_i\cdot \mathbf{C}_i]\in \mathbb{Z}_q^{n \times (m+nk)}$,
					\item $\textbf{e} \in \mathbb{Z}^{m+nk} \leftarrow \textsf{SampleD}(\textbf{T}_s. \textbf{A}_{s,\textbf{h}},\textbf{u},\sigma_2)$
				
					\end{enumerate}
			\item     $\textbf{c}_1=\overline{\textbf{c}}_1+ \mu \cdot \lfloor q/2\rfloor \in \mathbb{Z}_q^{\ell}, \quad \textbf{c}'_1=\overline{\textbf{c}}'_1+ H(\mu) \cdot \lfloor q/2\rfloor \in \mathbb{Z}_q^{\ell}$.
			\item Output the ciphertext $ct=(\textbf{c}_0, \textbf{c}_1, \textbf{r}_e,\textbf{r}_s,\textbf{c}'_0, \textbf{c}'_1,\textbf{r}'_e,   \textbf{e})$.
				\end{enumerate}

	\item[ \underline{\textsf{USC}($sk_r, pk_s,ct$)}:] On input a sender's public key  $ pk_s:=(\mathbf{A}_s, \mathbf{A}'_s)$, a receiver's private key $sk_r:=(\textbf{T}_r, \textbf{T}'_r)$, a ciphertext  $ct=(\textbf{c}_0, \textbf{c}_1, \textbf{r}_e,\textbf{r}_s,\textbf{c}'_0,$ $ \textbf{c}'_1,\textbf{r}'_e,   \textbf{e})$, do the following:
			\begin{enumerate}
	
				\item Compute  $\textbf{t}=f_{\overline{\textbf{A}}_r}(H_1(\textbf{A}_s))+f_{\textbf{B}}(\textbf{r}_e)\in \mathbb{Z}_q^{n}$ and  $\textbf{A}_{r,\textbf{t}}=[\overline{\mathbf{A}}_r|H_2(\textbf{t})\mathbf{G}-\overline{\mathbf{A}}_r\cdot \mathbf{T}_r]\in \mathbb{Z}_q^{n \times m}$.
				\item $(\textbf{s},\textbf{x}_0) \leftarrow \textsf{Invert}(\textbf{T}_r,\textbf{A}_{r,\textbf{t}}, \textbf{c}_0)$.
				\item Compute $\textbf{E} \in \mathbb{Z}^{m\times \ell} \leftarrow \textsf{SampleD}(\textbf{T}_r, \textbf{A}_{r,\textbf{t}},\textbf{U},\sigma_2)$.
				
			%	Compute $\textbf{E}\in \mathbb{Z}^{m \times \ell}$ such that $\textbf{A}_{r, \textbf{t}}\textbf{E}=\textbf{U}\!\! \mod q$ using \textsf{SampleD}
				\item Compute $\textbf{v}^t=\textbf{c}_1^t-(\textbf{c}_0-\textbf{x}_{0})^t\textbf{E}= \textbf{x}_1^t+\mu \cdot \lfloor q/2\rfloor$.
				\item Recover $\mu$ from $\textbf{v} \!\! \mod q$.
				\item  $\overline{\textbf{c}}_1=\textbf{c}_1- \mu \cdot \lfloor q/2\rfloor \!\! \mod q$, $\overline{\textbf{c}}'_1=\textbf{c}'_1- H(\mu) \cdot \lfloor q/2\rfloor \!\! \mod q$, and let $\overline{ct}:=(\textbf{c}_0, \overline{\textbf{c}}_1, \textbf{r}_e, \textbf{c}'_0, \overline{\textbf{c}}'_1, \textbf{r}'_e)$.

				\item Compute  $\textbf{h}=(h_1, \cdots, h_n)=f_{\overline{\textbf{A}}_s}(H_3(\mu| pk_r|\overline{ct}))+f_{\textbf{B}}(\textbf{r}_s)\in \mathbb{Z}_q^{n}$.
				\item $\textbf{A}_{s,\textbf{h}}=[\mathbf{A}_s|\mathbf{C}_0+\sum_{i=1}^{n}h_i\cdot \mathbf{C}_i]\in \mathbb{Z}_q^{n \times (m+nk)}$.
				\item If $\textbf{A}_{s,\textbf{h}}\cdot \textbf{e}=\textbf{u} \!\! \mod q$ and $\|\textbf{e}\|\leq \sigma_2\sqrt{m+nk}$ then 
				output $\mu$; otherwise, output $\bot$.
			\end{enumerate}
	\item [\underline{\textsf{Tag}($sk_r)$}:] On input a receiver's private key $sk_r:=(\textbf{T}_{r}, \textbf{T}'_{r})$,  return the tag $ tg_r:=\textbf{T}'_{r}$.
	
		\item[ \underline{\textsf{Test}($(tg_{r,i},ct_i),(tg_{r,j},ct_j)$)}:] On input a tag $ tg_{r,i}:=\textbf{T}'_{r,i}$, a ciphertext $ct_i=(\textbf{c}_{0,i}, \textbf{c}_{1,i}, \textbf{r}_{e, i}, \textbf{r}_{s, i},  \textbf{c}'_{0,i}, \textbf{c}'_{1,i}, \textbf{r}'_{e, i},  \textbf{e}_{i})$ with respect to the receiver $\mathcal{R}_i,$  and a tag $ tg_{r,j}:=\textbf{T}'_{r,j}$, a ciphertext $ct_j=(\textbf{c}_{0,j}, \textbf{c}_{1,j}, \textbf{r}_{e, j}, \textbf{r}_{s, j}, \textbf{c}'_{0,j}, \textbf{c}'_{1,j}, \textbf{r}'_{e, j},  \textbf{e}_{j})$ with respect to the receiver  $\mathcal{R}_j$, do the following:
		\begin{itemize}
			\item For $\mathcal{R}_i$, do:
				\begin{enumerate}
		
					\item Compute  $\textbf{t}'_i=f_{\overline{\textbf{A}}'_{r,i}}(H_1(\textbf{A}'_{s,i}))+f_{\textbf{B}'}(\textbf{r}'_{e,i})\in \mathbb{Z}_q^{n}$ and  $\textbf{A}'_{r,\textbf{t}_i}=[\overline{\mathbf{A}}'_{r,i}|H_2(\textbf{t}'_i)\mathbf{G}-\overline{\mathbf{A}}'_{r,i}\cdot \mathbf{T}'_{r,i}]\in \mathbb{Z}_q^{n \times m}$,
					\item $(\textbf{s}'_i,\textbf{x}'_{0, i}) \leftarrow \textsf{Invert}(\textbf{T}'_{r,i},\textbf{A}'_{r,\textbf{t}_i}, \textbf{c}'_{0,i})$.
					\item 
					
					Compute $\textbf{E}' \in \mathbb{Z}^{m\times \ell} \leftarrow \textsf{SampleD}(\textbf{T}'_r, \textbf{A}'_{r,\textbf{t}},\textbf{U}',\sigma_2)$.
					%Compute $\textbf{E}'_{i}\in \mathbb{Z}^{m \times \ell}$ s.t. $\textbf{A}'_{r, \textbf{t}_i}\textbf{E}'_{i}=\textbf{U}'\!\! \mod q$ using \textsf{SampleD}
					\item Compute $(\textbf{v}'_{i})^t=(\textbf{c}'_{1,i})^t-(\textbf{c}'_{0,i}-\textbf{x}'_{0,i})^t\textbf{E}_i= (\textbf{x}_{1,i}')^t+H(\mu_i)\cdot \lfloor q/2\rfloor $.
					\item Recover $H(\mu_i)$ from $\textbf{v}'_i\!\! \mod q$.	
				\end{enumerate}
				\item For $\mathcal{R}_j$: Do the same steps as above for $\mathcal{R}_i$ to recover $H(\mu_j)$.
				\item  Output 1 if  $H(\mu_i)=H(\mu_j)$. Otherwise, output $0$.
		\end{itemize}
	\end{description}

	\begin{theorem}[Correctness] \label{corrtheorem}
		The proposed $\mathsf{SCET}$ scheme is correct following the conditions mentioned in Subsection \ref{cor} provided that $H$ is collision-resistant.
	\end{theorem}
	\begin{proof}
			For any $pp \leftarrow \textsf{Setup}(1^{\lambda})$, $(pk_r, sk_r) \leftarrow \textsf{KGr} (pp)$, $(pk_s, sk_s) \leftarrow \textsf{KGs} (pp)$,  $(pk_{r_1}, sk_{r_1}) $ $ \leftarrow \textsf{KGr} (pp)$, $(pk_{s_1}, sk_{s_1}) \leftarrow \textsf{KGs} (pp)$, $(pk_{r_2}, sk_{r_2}) \leftarrow \textsf{KGr} (pp)$, $(pk_{s_2}, sk_{s_2}) \leftarrow \textsf{KGs} (pp)$, $tg_1 \leftarrow \textsf{Tag}(sk_{r_1})$ and $tg_2 \leftarrow \textsf{Tag}(sk_{r_2})$, any ciphertexts $ct_1$ and $ct_2$, and any message $\mu \in \mathcal{M}$. We need to check the following:
			\begin{itemize}
				\item First, we will prove that $\Pr[\mu=\textsf{USC}(sk_r, pk_s,\textsf{SC}( pk_r, sk_s,\mu))]=1-\textsf{negl}(\lambda).$ Indeed, let $ct=(\textbf{c}_0, \textbf{c}_1, \textbf{r}_e,\textbf{r}_s,\textbf{c}'_0,$ $ \textbf{c}'_1,\textbf{r}'_e,   \textbf{e})$ be a ciphertext outputted by $\textsf{SC}( pk_r, sk_s,\mu)$. Now what we need to verify is Step 5 in the \textsf{USC} algorithm. To succesfully recover $\mu=(\mu_1, \cdots, \mu_\ell)$ from $\textbf{v}=\textbf{x}_1+\mu \cdot \lfloor q/2\rfloor $, we compare each component of $\textbf{v}=(v_1, \cdots, v_\ell)$ to  $q/2$. If $|v_i|<q/2$ then $\mu_i=0$. Otherwise, $\mu_i=1$. This is thanks to the smallness of $\textbf{x}_1\leftarrow D_{\mathbb{Z}^{\ell}, \alpha q}$.
				\item Second, we need to show that if
				 $\textsf{USC}(sk_{r_1}, pk_{s_1},ct_1)$ $=\textsf{USC}(sk_{r_2}, pk_{s_2},ct_2)=\mu \neq \bot,$
				then $$\Pr[\textsf{Test}(tg_1,ct_1,tg_2,ct_2) =1]=1-\textsf{negl}(\lambda).$$ This can be done in the same way as above.
				\item Finally, we show that if
				$\textsf{USC}(sk_{r_1}, pk_{s_1},ct_1)=\mu_1 \neq \textsf{USC}(sk_{r_2}, pk_{s_2},ct_2)=\mu_2,$
				then $$\Pr[\textsf{Test}(tg_1,ct_1,tg_2,ct_2)=1]=\textsf{negl}(\lambda).$$
					This can be done in the same way as above with noting that if $H(\mu_1)=H(\mu_2)$ happens, then it must be that $\mu_1=\mu_2$ due to the collision-resistance of $H$.
			\end{itemize}
	Note that, the negligibility in the above conditions comes from that of the trapdoor algotihms being used such as $\textsf{Invert}, \textsf{SampleD}$ with appropriately chosen parameters.

	\end{proof}

	\section{Security Analysis} \label{security}
	
	\begin{theorem}[IND-iCCA1] \label{indtheorem}
	The proposed $\mathsf{SCET}$ scheme is IND-iCCA1 secure under the hardness of the decisional-LWE $\mathsf{dLWE}_{n,2(\overline{m}+\ell), q,\alpha q}$ problem and the collision-resistance of the functions $ f_{\overline{\textbf{A}}_{r}}(\cdot)+f_{\textbf{B}}(\cdot)$ for any $\overline{\mathbf{A}}_r \xleftarrow{\$} \mathbb{Z}_q^{n \times\overline{ m}}$ and any $\mathbf{B} \xleftarrow{\$} \mathbb{Z}_q^{n \times m}$.
	\end{theorem}
	
	\begin{proof}  We consider a sequence of games in which the first game \textbf{Game IND0} is the original one. And the last game \textbf{Game IND5} is the ``uniform-based" game. We will demnostrate that \textbf{Game IND}$i$ and \textbf{Game IND}$(i+1)$ are indistinguishable for $i\in \{0, \cdots, 4\}$.
		
		\begin{description}
			\item[Game IND0.] This is the original IND-iCCA1  game. Suppose that the target receiver is $r^*$ and the target sender announced at \textbf{Challenge} phase by the adversary $\mathcal{A}_1$ is $s^*$. Also, let  $\textbf{t}^*\gets f_{\overline{\textbf{A}}_{r^*}}(H_1(\textbf{A}_{s^*}))+f_{\textbf{B}}(\textbf{r}^*_e)$, and $\textbf{t}'^*\gets f_{\overline{\textbf{A}}'_{r^*}}(H_1(\textbf{A}'_{s^*}))+f_{\textbf{B}'}(\textbf{r}'^*_e)$.
			 \item[Game IND1.]  This game is same as \textbf{Game IND0}, except that if  $\mathcal{A}_1$ makes an unsigcryption query  $(r^*, s,ct)$ such that $f_{\overline{\textbf{A}}_{r^*}}(H_1(\textbf{A}_s))+f_{\textbf{B}}(\textbf{r}_e)=\textbf{t}^*$ or $f_{\overline{\textbf{A}}'_{r^*}}(H_1(\textbf{A}'_s))+f_{\textbf{B}'}(\textbf{r}'_e)=\textbf{t}'^*$, where $\textbf{t}^*$ and $\textbf{t}'^*$ are defined as in \textbf{Game IND0} (we name this event by $\textsf{Event}_1$), then the challenger outputs $\bot$. 
			\item   \textbf{Game IND1} and \textbf{Game IND0} are indistinguishable since the probability that the event $\textsf{Event}_1$ happens is negligible due to the collision resistance of  $ f_{\overline{\textbf{A}}_{r^*}}(\cdot)+f_{\textbf{B}}(\cdot)$, and $ f_{\overline{\textbf{A}}'_{r^*}}(\cdot)+f_{\textbf{B}'}(\cdot)$. 
	
			  \item[Game IND2.]  This game is same as \textbf{Game IND1}, except that  instead of  $\mathbf{B}, \mathbf{B}'$ being uniform in $ \mathbb{Z}_q ^{n \times m}$, use $\textsf{GenTrap}(n, \overline{m},q, \sigma_1)$ to generate  $(\mathbf{B}, \mathbf{T}_{\textbf{B}}), (\mathbf{B}', \mathbf{T}_{\textbf{B}}') \in  \mathbb{Z}_q^{n \times m} \times \mathbb{Z}_q^{\overline{m} \times nk}$.
			  		
		\item  \textbf{Game IND2} and \textbf{Game IND1} are indistinguishable due to the property of $\textsf{GenTrap}$ algorithm. Namely, although being genereted using $\textsf{GenTrap}$, both $\mathbf{B}, \mathbf{B}'$ look uniform in $ \mathbb{Z}_q ^{n \times m}$.   

	  \item[Game IND3.] This game is same as \textbf{Game IND2}, except that in the \textbf{Setup} phase, for the target receiver $r^*$, the challenger generates as follows:
				  
				  	  	\begin{enumerate}
				  			 
				  			 \item Choose $\textbf{t}^*, \textbf{t}'^*\in \mathbb{Z}_q^n$ uniformly at random. The challenger uses  $\textbf{t}^*, \textbf{t}'^*$  to build $\mathbf{A}_{r^*}, \mathbf{A}'_{r^*}$.
				  			\item  Choose $ \mathbf{T}_{r^*}, \mathbf{T}'_{r^*} \leftarrow D_{\mathbb{Z}^{\overline{m} \times n k }, \sigma_1} $ and then set $\mathbf{A}_{r^*}=[\overline{\mathbf{A}}_{r^*}|-H_2(\textbf{t}^*)\mathbf{G}-\overline{\mathbf{A}}_{r^*}\cdot \mathbf{T}_{r^*}] \in \mathbb{Z}_q^{n \times m}$, $\mathbf{A}'_{r^*}=[\overline{\mathbf{A}}'_{r^*}|-H_2(\textbf{t}'^*)\mathbf{G}-\overline{\mathbf{A}}'_{r^*}\cdot \mathbf{T}'_{r^*}] \in \mathbb{Z}_q^{n \times m}$.
				  			\item The public key for $r^*$ is $pk_{r^*}=(\mathbf{A}_{r^*},\mathbf{A}'_{r^*})$ and the private key key for $r^*$ is $sk_{r^*}=(\mathbf{T}_{r^*},\mathbf{T}'_{r^*})$.
				  		 	  	\end{enumerate}

	\item  \textbf{Game IND3} and \textbf{Game IND2} are indistinguishable since the distribution of  $\mathbf{A}_{r^*}, \mathbf{A}'_{r^*}$ is the same as that of $\mathbf{A}_{r}, \mathbf{A}'_{r}$ for all $r \neq r^*$ which are generated in Step 1 of the \textsf{KG} algorithm.  		
			  		  
			  \item[Game IND4.] This game is the same as \textbf{Game IND3}, except that  in the \textbf{Challenge} phase, the challenger performs the following:
			  	\begin{enumerate}
			
			  \item Choose randomly $b \xleftarrow{\$} \{0,1\}$.

			\item $\mathbf{A}_{r^*,\textbf{t}^*}=\mathbf{A}_{r^*}+[\textbf{0}|H_2(\textbf{t}^*)\mathbf{G}]=[\overline{\mathbf{A}}_{r^*}|-\overline{\mathbf{A}}_{r^*}\cdot \mathbf{T}_{r^*}] \in \mathbb{Z}_q^{n \times m}$,\\ $\mathbf{A}'_{r^*,\textbf{t}'^*}=\mathbf{A}'_{r^*}+[\textbf{0}|H_2(\textbf{t}'^*)\mathbf{G}]=[\overline{\mathbf{A}}_{r^*}|-\overline{\mathbf{A}}'_{r^*}\cdot \mathbf{T}'_{r^*}] \in \mathbb{Z}_q^{n \times m}$.
		\item Compute $\textbf{r}_e^* \leftarrow \textsf{SampleD}(\textbf{T}_\textbf{B}, \textbf{B}, (\textbf{t}^*-f_{\overline{\textbf{A}}_{r^*}}(H_1(\textbf{A}_{s^*}))),\alpha q)$,  ${\textbf{r}'_e}^* \leftarrow \textsf{SampleD}(\textbf{T}'_\textbf{B}, \textbf{B}', (\textbf{t}'^*-f_{\overline{\textbf{A}}'_{r^*}}(H_1(\textbf{A}'_{s^*}))),\alpha q)$. 		 
	
		\item Sample $\textbf{s}, \textbf{s}' \xleftarrow{\$} \mathbb{Z}_q^{n}$, \quad  $\hat{\textbf{x}}_0, \hat{\textbf{x}}_0'\leftarrow D_{\mathbb{Z}^{\overline{m}}, \alpha q}$,\quad  $\textbf{x}_1, \textbf{x}_1'\leftarrow D_{\mathbb{Z}^{\ell}, \alpha q}$.
			  	\item Compute $\hat{\textbf{c}}_0=\textbf{s}^t\overline{\mathbf{A}}_{r^*}+\hat{\textbf{x}}_0^t \in \mathbb{Z}_q^{\overline{m}}$, $\overline{\textbf{c}}_1=\textbf{s}^t\textbf{U}+ \textbf{x}^t_1 \in \mathbb{Z}_q^{\ell}$,
			 	$\hat{\textbf{c}}'_0=(\textbf{s}')^t\overline{\mathbf{A}}_{r^*}'+(\hat{\textbf{x}}_0')^t \in \mathbb{Z}_q^{\overline{m}}$, $\overline{\textbf{c}}'_1=(\textbf{s}')^t\textbf{U}'+ (\textbf{x}'_1)^t \in \mathbb{Z}_q^{\ell}$,
			  	%	\item Compute ${\textbf{r}_e}^* \leftarrow \textsf{SampleD}(\textbf{T}_\textbf{B}, \textbf{B}, (\textbf{t}^*-f_{\overline{\textbf{A}}_{r^*}}(\textbf{A}_{s^*})))$,\\  ${\textbf{r}'_e}^* \leftarrow \textsf{SampleD}(\textbf{T}'_\textbf{B}, \textbf{B}', (\textbf{t}'^*-f_{\overline{\textbf{A}}'_{r^*}}(\textbf{A}'_{s^*})))$  
			  	%	\item $\mathbf{A}_{r^*,\textbf{t}^*}=\mathbf{A}_{r^*}+[\textbf{0}|H_2(\textbf{t}^*)\mathbf{G}]=[\overline{\mathbf{A}}_{r^*}|-\overline{\mathbf{A}}_{r^*}\cdot \mathbf{T}_{r^*}] \in \mathbb{Z}_q^{n \times m}$,\\ $\mathbf{A}'_{r^*,\textbf{t}'^*}=\mathbf{A}'_{r^*}+[\textbf{0}|H_2(\textbf{t}'^*)\mathbf{G}]=[\overline{\mathbf{A}}_{r^*}|-\overline{\mathbf{A}}'_{r^*}\cdot \mathbf{T}'_{r^*}] \in \mathbb{Z}_q^{n \times m}$.
		\item Set  ${(\textbf{c}^*_0)}^t:=(\hat{\textbf{c}}_0^t|\hat{\textbf{c}}_0^t\textbf{T}_{r^*})$, \quad  ${(\textbf{c}'^*_0)}^t:=((\hat{\textbf{c}}'_0)^t|(\hat{\textbf{c}}'_0)^t\textbf{T}'_{r^*})$.
	%	\item  Sample $\textbf{r}^*_s, \textbf{r}'^*_s \leftarrow D_{\mathbb{Z}^{m }, \alpha q}$, and compute  the challenge ciphertext ${ct}^*=({\textbf{c}_0}^*, {\textbf{c}_1}^*, {\textbf{r}_e}^*,{\textbf{c}'_0}^*, {\textbf{c}'_1}^*,{\textbf{r}'_e}^*, {\textbf{e}}^*)$ using Steps 6-8 of \textsf{SC}($pk_{r^*},sk_{s^*},\mu^*_b$) with $\overline{ct}^*=(\overline{\textbf{c}}^*_0, \overline{\textbf{c}}^*_1, \textbf{r}^*_e)$. 
	
		\item Sign on $\mu^*_b| pk_{r^*}|\overline{ct}^*$ with $\overline{ct}^*=(\textbf{c}^*_0, \overline{\textbf{c}}^*_1, \textbf{r}^*_e, \textbf{c}'^*_0, \overline{\textbf{c}}'^*_1, \textbf{r}'^*_e)$ to get the signature $(\textbf{e}^*, \textbf{r}_s^*)$ as usuall. 
		
		\iffalse  follows:					
						\begin{enumerate}
						
		\item Sample $\textbf{r}^*_s\leftarrow D_{\mathbb{Z}^{m }, \alpha q}$.	%Sample $\textbf{r} \leftarrow D_{\mathbb{Z}^{\overline{m} }, \alpha q}$ and let $(\textbf{r}^*_s)^t=(\textbf{r}^t|\textbf{r}^t\textbf{T}_{r^*})$.
		\item $\textbf{h}^*=(h^*_1, \cdots, h^*_n)=f_{\overline{\textbf{A}}_{s^*}}(H_3(\mu^*_b| pk_{r^*}|\overline{ct}^*))+f_{\textbf{B}}(\textbf{r}^*_{s})\in \mathbb{Z}_q^{n}$.
		\item $\textbf{A}_{s^*,\textbf{h}^*}=[\mathbf{A}_{s^*}|\mathbf{C}_0+\sum_{i=1}^{n}h^*_i\cdot \mathbf{C}_i]\in \mathbb{Z}_q^{n \times (m+nk)}$.
		\item $\textbf{e}^* \in \mathbb{Z}^{m+nk} \leftarrow \textsf{SampleD}(\textbf{T}_{s^*}, \textbf{A}_{s^*,\textbf{h}^*},\textbf{u},\sigma_2)$.
								\end{enumerate}
								\fi	
		\item  $\textbf{c}^*_1=\overline{\textbf{c}}^*_1+ \mu^*_b \cdot \lfloor q/2\rfloor$, \quad  $\textbf{c}'^*_1=\overline{\textbf{c}}'^*_1+ H(\mu^*_b) \cdot \lfloor q/2\rfloor$	
		 
				\item Return $ct^*=(\textbf{c}^*_0, \textbf{c}^*_1, \textbf{r}^*_{e},\textbf{r}^*_{s},\textbf{c}'^*_0, \textbf{c}'^*_1,\textbf{r}'^*_e, \textbf{e}^*)$ to $\mathcal{A}_1$.
			  	\end{enumerate}

			\item  \textbf{Game IND4} and \textbf{Game IND3} are indistinguishable as the challenger ís just following the real signcryption algorithm $\textsf{SC}$ with  $\mathbf{A}_{r^*}=[\overline{\mathbf{A}}_{r^*}|-H_2(\textbf{t}^*)\mathbf{G}-\overline{\mathbf{A}}_{r^*}\cdot \mathbf{T}_{r^*}] \in \mathbb{Z}_q^{n \times m}$, $\mathbf{A}'_{r^*}=[\overline{\mathbf{A}}'_{r^*}|-H_2(\textbf{t}'^*)\mathbf{G}-\overline{\mathbf{A}}'_{r^*}\cdot \mathbf{T}'_{r^*}] \in \mathbb{Z}_q^{n \times m}$ and the distribution of $\textbf{r}^*_e, \textbf{r}'^*_e $ is still $ D_{\mathbb{Z}^{m }, \alpha q}$ by the property of \textsf{SampleD}.   	  			 		
			  			 		
			   \item[Game IND5.] This game is the same as \textbf{Game 4}, except that $\overline{\textbf{c}}^*_0, \overline{\textbf{c}}^*_1, \textbf{r}^*_e, \overline{\textbf{c}}'^*_0, \overline{\textbf{c}}'^*_1, \textbf{r}'^*_e$ are  chosen uniformly at random. 
			  
			\item   Below, we are going to show that \textbf{Game IND5} and \textbf{Game IND4} are indistinguishable using a  reduction from the hardness of the decision LWE problem.
			   \item[Reduction from LWE.] 
			  Suppose that $\mathcal{A}_1$ can distinguish \textbf{Game IND5} and \textbf{Game IND4}. Then we will construct an algorithm  $\mathcal{B}_1$ that can solve an LWE instance.

			  \textbf{LWE Instance.} $\mathcal{B}_1$ is given a pair $(\textbf{F},\textbf{c}^t) \in \mathbb{Z}_q^{n \times (2\overline{ m}+2\ell)} \times \mathbb{Z}_q^{2(\overline{m}+\ell)}$ that can be parsed as $(\mathbf{A}|\mathbf{A}'|\textbf{U}|\textbf{U}',\hat{\textbf{c}}_0^t|(\hat{\textbf{c}}'_0)^t|\overline{\textbf{c}}_1^t|(\overline{\textbf{c}}'_1)^t) \in \mathbb{Z}_q^{n \times (\overline{ m}+\overline{ m}+\ell+\ell)} \times \mathbb{Z}_q^{\overline{ m}+\overline{ m}+\ell+\ell}$, and $\mathcal{B}_1$ has to decide whether \begin{itemize}
			  \item (i) $(\textbf{F},\textbf{c}^t)$ is an LWE instance:  $\hat{\textbf{c}}_0=\textbf{s}^t\textbf{A}+ \hat{\textbf{x}}^t_0 \in \mathbb{Z}_q^{\overline{m}}$, $\overline{\textbf{c}}_1=\textbf{s}^t\textbf{U}+ \textbf{x}^t_1 \in \mathbb{Z}_q^{\ell}$,
			$\hat{\textbf{c}}'_0=(\textbf{s}')^t\textbf{A}'+ (\hat{\textbf{x}}_0')^t \in \mathbb{Z}_q^{\overline{m}}$, $\overline{\textbf{c}}'_1=(\textbf{s}')^t\textbf{U}'+ (\textbf{x}'_1)^t \in \mathbb{Z}_q^{\ell}$, for some $\textbf{s}, \textbf{s}' \xleftarrow{\$} \mathbb{Z}_q^{n}$, $\hat{\textbf{x}}_0, \hat{\textbf{x}}_0'\leftarrow D_{\mathbb{Z}^{\overline{m}}, \alpha q}$, $\textbf{x}_1, \textbf{x}_1'\leftarrow D_{\mathbb{Z}^{\ell}, \alpha q}$; or 
			\item (ii) $(\textbf{F},\textbf{c}^t)$ is uniform in $\mathbb{Z}_q^{n \times (2\overline{ m}+2\ell)} \times \mathbb{Z}_q^{2(\overline{m}+\ell)}$.
			  \end{itemize}
			  
			  The algorithms  $\mathcal{B}_1$ and $\mathcal{A}_1$ play the following game:
			  
			  \textbf{Setup.}   $\mathcal{B}_1$ simulates public parameters $pp$, public keys for $M$ senders and $N$ receivers as follows:
			  \begin{itemize}
			  	\item Pick $n,q, k, \overline{m}, m,\ell, n, N, M, \alpha, \sigma_1, \sigma_2$ and use hash functions $ H, H_1,$  $  H_2, H_3$. The message space is $\mathcal{M}$.
			  	
			  	\item Randomly guess $r^* \xleftarrow{\$} \{1, \cdots, N\}$ to be the target receiver targeted by $\mathcal{A}_1$, and then set $\overline{\mathbf{A}}_{r^*}:=\textbf{A}$, $\overline{\mathbf{A}}'_{r^*}:=\textbf{A}'$. 
			  	%	\item Use $\textsf{TrapGen}(n, m,q)$ to generate  $(\mathbf{A}'_{s^*}, \mathbf{T}'_{s^*}) \in  \mathbb{Z}_q^{n \times m} \times \mathbb{Z}_q^{m \times m}$.
			  	\item Choose $\textbf{t}^*, \textbf{t}'^*\in \mathbb{Z}_q^n$ uniformly at random and choose $ \mathbf{T}_{r^*}, \mathbf{T}'_{r^*} \leftarrow D_{\mathbb{Z}^{\overline{m} \times n k }, \sigma_1} $ and then set $\mathbf{A}_{r^*}=[\overline{\mathbf{A}}_{r^*}|-H_2(\textbf{t}^*)\mathbf{G}-\overline{\mathbf{A}}_{r^*}\cdot \mathbf{T}_{r^*}] \in \mathbb{Z}_q^{n \times m}$, $\mathbf{A}'_{r^*}=[\overline{\mathbf{A}}'_{r^*}|-H_2(\textbf{t}'^*)\mathbf{G}-\overline{\mathbf{A}}'_{r^*}\cdot \mathbf{T}'_{r^*}] \in \mathbb{Z}_q^{n \times m}$.
			  	\item Also, use $\textsf{GenTrap}(n, \overline{m},q, \sigma_1)$ to generate  $(\mathbf{B}, \mathbf{T}_{\textbf{B}}), (\mathbf{B}', \mathbf{T}_{\textbf{B}}') \in  \mathbb{Z}_q^{n \times m} \times \mathbb{Z}_q^{\overline{m} \times nk}$.
			 % 	\item Choose  $\mathbf{U},\mathbf{U}' \in \mathbb{Z}_q ^{n \times \ell}$.  $\textcolor{red}{\mathbf{B}, \mathbf{B}'} \in \mathbb{Z}_q ^{n \times m}$,
			  	\item Sample $ \mathbf{u} \xleftarrow{\$} \mathbb{Z}_q ^{n}$ and  matrices $\mathbf{C}_0, \cdots,$ $ \mathbf{C}_{n}$, $\mathbf{C}'_0, \cdots,$ $ \mathbf{C}'_{n} \xleftarrow{\$} \mathbb{Z}_q^{n \times m}$.
			  %	\item For $i \in \{0, \cdots, n\}$, set $\textbf{C}_i:=\textbf{A}\textbf{R}_i+\gamma_i\textbf{B} \text{ (mod } q)$, $\textbf{C}'_i:=\textbf{A}'\textbf{R}'_i+\gamma'_i\textbf{B}' \text{ (mod } q)$.
			  	\item For each receiver $r \in [N]\setminus \{r^*\}$ and each sender $s\in [M]$, use the algorithm \textsf{KG} to generate  $(\mathbf{A}_r, \mathbf{T}_r)$, $ (\mathbf{A}'_r, \mathbf{T}'_r)$, $(\mathbf{A}_s, \mathbf{T}_s)$, $(\mathbf{A}'_s, \mathbf{T}'_s) \in  \mathbb{Z}_q^{n \times m} \times \mathbb{Z}_q^{\overline{m} \times n k}$.

			  	\item Set $pp=\{n, q, k,  \overline{m},, m,\ell, n, \alpha, \sigma_1,  \sigma_2, N, M,  \mathcal{M}, (\bf{C}_i, \bf{C}'_i)_{i=0}^{n}, H, H_1,$ $ H_2, $ $H_3\}$ as public parameters and  $pk_{r}=(\mathbf{A}_r, \mathbf{A}'_r)$, $pk_{s}=(\mathbf{A}_s, \mathbf{A}'_s)$ as public keys corresponding to each receiver $r \in [N]$, and each sender $s\in [M]$.
			  	\item Send $pp$, $pk_s$'s, $pk_r$'s all to the adversary  $\mathcal{A}_1$.
			  \end{itemize}
			 
			  \textbf{Phase 1.}  $\mathcal{A}_1$ adaptively makes a polynomially bounded number of the following queries:
			\begin{itemize}
		\item Private key query  PKQ$(r)$: If $r=r^*$, $\mathcal{B}_1$ rejects the query. Otherwise, $\mathcal{B}_1$   returns the private key $sk_{r}=(\textbf{T}_r, \textbf{T}'_r)$  to $\mathcal{A}_1$. 
		\item Signcryption query SCQ$(r,s,\mu)$: $\mathcal{B}_1$ sends the output $ct$ of $\textsf{SC}(pk_{r}$, $sk_{s},\mu)$ back to $\mathcal{A}_1$.
		\item Unsigncryption query   USQ$(r,s,ct)$: If  $\mathcal{A}_1$ makes an unsigcryption query  $(r, s,ct)$ such that $\textbf{t}^*=f_{\overline{\textbf{A}}_r}(H_1(\textbf{A}_s))+f_{\textbf{B}}(\textbf{r}_e)$ and $\textbf{t}'^*=f_{\overline{\textbf{A}}'_r}(H_1(\textbf{A}'_s))+f_{\textbf{B}'}(\textbf{r}'_e)$ then  $\mathcal{B}_1$ outputs $\bot$. Otherwise, $\mathcal{B}_1$ sends the output $\mu$/$\bot$ of $\textsf{USC}(sk_{r},pk_{s},ct)$ back to $\mathcal{A}_1$. 
		\item Tag query TGQ$(r)$: If $r=r^*$, $\mathcal{B}_1$ rejects the query. Otherwise, $\mathcal{B}_1$ sends the output $tg_{r}=\textbf{T}'_r$ back to $\mathcal{A}_1$. 
			\end{itemize}
			 				
	\textbf{Challenge.} $\mathcal{A}_1$ submits two messages $\mu^*_0, \mu^*_1$ together with the target sender's keys $(pk_{s^*},sk_{s^*})$. The adversary $\mathcal{B}_1$ does the following:
		\begin{enumerate}
	\item Choose randomly $b \xleftarrow{\$} \{0,1\}$.
	\item Compute $\textbf{r}_e^* \leftarrow \textsf{SampleD}(\textbf{T}_\textbf{B}, \textbf{B}, (\textbf{t}^*-f_{\overline{\textbf{A}}_{r^*}}(H_1(\textbf{A}_{s^*}))),\alpha q)$,  ${\textbf{r}'_e}^* \leftarrow \textsf{SampleD}(\textbf{T}'_\textbf{B}, \textbf{B}', (\textbf{t}'^*-f_{\overline{\textbf{A}}'_{r^*}}(H_1(\textbf{A}'_{s^*}))),\alpha q)$.  
	\item $\mathbf{A}_{r^*,\textbf{t}^*}=\mathbf{A}_{r^*}+[\textbf{0}|H_2(\textbf{t}^*)\mathbf{G}]=[\overline{\mathbf{A}}_{r^*}|-\overline{\mathbf{A}}_{r^*}\cdot \mathbf{T}_{r^*}] \in \mathbb{Z}_q^{n \times m}$,\\ $\mathbf{A}'_{r^*,\textbf{t}'^*}=\mathbf{A}'_{r^*}+[\textbf{0}|H_2(\textbf{t}'^*)\mathbf{G}]=[\overline{\mathbf{A}}_{r^*}|-\overline{\mathbf{A}}'_{r^*}\cdot \mathbf{T}'_{r^*}] \in \mathbb{Z}_q^{n \times m}$.
	\item Set  ${(\textbf{c}^*_0)}^t:=(\hat{\textbf{c}}_0^t|\hat{\textbf{c}}_0^t\textbf{T}_{r^*})$, \quad  ${(\textbf{c}'^*_0)}^t:=((\hat{\textbf{c}}'_0)^t|(\hat{\textbf{c}}'_0)^t\textbf{T}'_{r^*})$.
%	\item  Sample $\textbf{r}^*_s, \textbf{r}'^*_s \leftarrow D_{\mathbb{Z}^{m }, \alpha q}$, and compute  the challenge ciphertext ${ct}^*=({\textbf{c}_0}^*, {\textbf{c}_1}^*, {\textbf{r}_e}^*,{\textbf{c}'_0}^*, {\textbf{c}'_1}^*,{\textbf{r}'_e}^*, {\textbf{e}}^*)$ using Steps 6-8 of \textsf{SC}($pk_{r^*},sk_{s^*},\mu^*_b$) with $\overline{ct}^*=(\overline{\textbf{c}}^*_0, \overline{\textbf{c}}^*_1, \textbf{r}^*_e)$. 

	\item Sign on $\mu^*_b| pk_{r^*}|\overline{ct}^*$ with $\overline{ct}^*=(\textbf{c}^*_0, \overline{\textbf{c}}^*_1, \textbf{r}^*_e, \textbf{c}'^*_0, \overline{\textbf{c}}'^*_1, \textbf{r}'^*_e)$ to get the signature $(\textbf{e}^*, \textbf{r}_s^*)$ as usuall. 
	
	\iffalse  follows:					
					\begin{enumerate}
					
	\item Sample $\textbf{r}^*_s\leftarrow D_{\mathbb{Z}^{m }, \alpha q}$.	%Sample $\textbf{r} \leftarrow D_{\mathbb{Z}^{\overline{m} }, \alpha q}$ and let $(\textbf{r}^*_s)^t=(\textbf{r}^t|\textbf{r}^t\textbf{T}_{r^*})$.
	\item $\textbf{h}^*=(h^*_1, \cdots, h^*_n)=f_{\overline{\textbf{A}}_{s^*}}(H_3(\mu^*_b| pk_{r^*}|\overline{ct}^*))+f_{\textbf{B}}(\textbf{r}^*_{s})\in \mathbb{Z}_q^{n}$.
	\item $\textbf{A}_{s^*,\textbf{h}^*}=[\mathbf{A}_{s^*}|\mathbf{C}_0+\sum_{i=1}^{n}h^*_i\cdot \mathbf{C}_i]\in \mathbb{Z}_q^{n \times (m+nk)}$.
	\item $\textbf{e}^* \in \mathbb{Z}^{m+nk} \leftarrow \textsf{SampleD}(\textbf{T}_{s^*}, \textbf{A}_{s^*,\textbf{h}^*},\textbf{u},\sigma_2)$.
							\end{enumerate}
							\fi	
	\item  $\textbf{c}^*_1=\overline{\textbf{c}}^*_1+ \mu^*_b \cdot \lfloor q/2\rfloor$, \quad  $\textbf{c}'^*_1=\overline{\textbf{c}}'^*_1+ H(\mu^*_b) \cdot \lfloor q/2\rfloor$	
	 
			\item Return $ct^*=(\textbf{c}^*_0, \textbf{c}^*_1, \textbf{r}^*_{e},\textbf{r}^*_{s},\textbf{c}'^*_0, \textbf{c}'^*_1,\textbf{r}'^*_e, \textbf{e}^*)$ to $\mathcal{A}_1$.
		
				\end{enumerate}

 	 \textbf{Phase 2.} $\mathcal{A}_1$ queries the oracles again as in \textbf{Phase 1} with a restriction that $\mathcal{A}_1$ is not allowed to make the queries PKQ$(r^*)$ and   USQ$(r^*,s^*, ct^*)$.
	 \textbf{Output.} $\mathcal{B}_1$ outputs whatever $\mathcal{A}_1$  outputs.
			 			
	\end{description}
	  \textbf{Analysis.} The probability that an unsigncryption query $(r,s,ct)$ makes $\textbf{t}=\textbf{t}^*$ and $\textbf{t}'=\textbf{t}'^*$ is negligible as $\textbf{t}^*$ and $\textbf{t}'^*$ are chosen randomly in \textbf{Setup} phase. Then, we have $H_2(\textbf{t}-\textbf{t}^*)$ and $H_2(\textbf{t}'-\textbf{t}'^*)$ are invertible then we can apply \textsf{Invert} as in the real unsigncryption algorithm \textsf{USC}.  Obviously, if $(\textbf{F},\textbf{c}^t)$ is the LWE instance then the view of $\mathcal{A}_1$ as in \textbf{Game IND4}; while if $(\textbf{F},\textbf{c}^t)$ is uniform in $\mathbb{Z}_q^{n \times (2\overline{ m}+2\ell)} \times \mathbb{Z}_q^{2(\overline{m}+\ell)}$ then the view of $\mathcal{A}_1$ as in \textbf{Game IND5}. Therefore, if $\mathcal{A}_1$ can distinguish \textbf{Game IND4} and \textbf{Game IND5} then $\mathcal{B}_1$ can solve the decision LWE problem. \qed	
		
	 %Indeed, we have the distribution of $[\mathbf{A}'_r||\mathbf{C}'_0]$, $[\mathbf{A}'_r||\mathbf{C}'_0]$ is negligibly close to uniform
	
	\end{proof}

		\begin{theorem}[OW-iCCA1] \label{owtheorem}
			The proposed $\mathsf{SCET}$ scheme is OW-iCCA1 secure provided that $H$ is an one-way hash function, the $\mathsf{dLWE}_{n,2(\overline{m}+\ell), q,\alpha q}$ problem is hard and the functions $ f_{\overline{\textbf{A}}_{r}}(\cdot)+f_{\textbf{B}}(\cdot)$ are collision-resistant for any $\overline{\mathbf{A}}_r \xleftarrow{\$} \mathbb{Z}_q^{n \times\overline{ m}}$ and any $\mathbf{B} \xleftarrow{\$} \mathbb{Z}_q^{n \times m}$. In particular, the advantage of the OW-iCCA1 advesary is 
			$$\epsilon \leq \epsilon_{H,OW}+ \epsilon_{f,CR}+\epsilon_{LWE},$$
		where	$\epsilon_{H,OW}$ is the advantage of breaking the one-wayness of $H$,  $\epsilon_{f,CR}$ is the advantage of  finding collision for any functions $ f_{\overline{\textbf{A}}_{r}}(\cdot)+f_{\textbf{B}}(\cdot)$ and $\epsilon_{LWE}$ is the advantage of solving the  $\mathsf{dLWE}_{n,2(\overline{m}+\ell), q,\alpha q}$ problem. 
			
		\end{theorem}
		
		\begin{proof} We prove by giving a sequence of five games in which the first game is the original OW-iCCA1 one and in the last game, the ciphertext will be chosen randomly. Obviously, in the last game  the advantage of the OW-iCCA1 adversary is zero.  For $i \in \{0,1,2,3,4\}$, let $W_i$ be the event that the OW-iCCA1 adversary $\mathcal{A}_2$ wins \textbf{Game OW$i$}, we need to prove that $\Pr[W_0] $ is negligible. To do that we will show that for $i\in \{0,\cdots, 4\}$, $|\Pr[W_{i}]-\Pr[W_{i+1}]| $ is negligible, guaranteed by the one-wayness of hash functions  and especially the hardness of the decision LWE problem.
		%  Assume that there is an OW-CCA1 adversary $\mathcal{A}_2$ successfully against the $\mathsf{SCET}$. Then we use $\mathcal{A}_2$  as a subroutine to construct an algorithm  $\mathcal{B}_2$ that can break the  IND-CCA1 security.	$\mathcal{B}_2$ receives $pp$ and all $pk_r={(\textbf{A}_r,\textbf{A}'_r)}$, $pk_s={(\textbf{A}_s,\textbf{A}'_s)}$ from the INDCCA2 game's challenger $\mathcal{C}$. 
			
			\begin{description}
				\item[Game OW0.] This is the original OW-iCCA1 game. Suppose that the target receiver is $r^*$ and the target sender announced at \textbf{Challenge} phase by the adversary $\mathcal{A}_2$ is $s^*$. Also, assume that $\textbf{t}^*:=f_{\overline{\textbf{A}}_{r^*}}(H_1(\textbf{A}_{s^*}))+f_{\textbf{B}}(\textbf{r}^*_e)$, and $\textbf{t}'^*:=f_{\overline{\textbf{A}}'_{r^*}}(H_1(\textbf{A}'_{s^*}))+f_{\textbf{B}'}(\textbf{r}'^*_e)$. Note that, the adversary $\mathcal{A}_2$ can get the trapdoor $\textbf{T}'_{r^*}$ using the trapdoor query for the target receiver $r^*$.
		
			\item[Game OW1.]  This game is same as \textbf{Game OW0}, except that  in the \textbf{Challenge} phase, on the challenge plaintext $\mu^*  \xleftarrow{\$}\mathcal{M}$, the challenger first chooses $\mu' \xleftarrow{\$}\mathcal{M}$ then signcrypts $\mu^*$  in $\textbf{c}_1^*$  and  $ H(\mu')$ instead of $\mu^*$ in ${\textbf{c}'}_1^*$ of the challenge ciphertext $ct^*$, i.e.,  $\textbf{c}^*_1=\overline{\textbf{c}}^*_1+ \mu^* \cdot \lfloor q/2\rfloor$,\quad  $\textbf{c}'^*_1=\overline{\textbf{c}}'^*_1+ H(\mu') \cdot \lfloor q/2\rfloor$	.
			
			\item 			 Since the view of the adverary $\mathcal{A}_2$ is the same  in both \textbf{Game OW1} and \textbf{Game OW0},  except the case $\mathcal{A}_2$ can break the one-wayness of $H$, we have $$|\Pr[W_1]-\Pr[W_0]|\leq \epsilon_{H,OW}.$$

				\item[Game OW2.]  This game is same as \textbf{Game OW1}, except that if  $\mathcal{A}_2$ makes an unsigcryption query  $(r^*, s,ct)$ such that $f_{\overline{\textbf{A}}_{r^*}}(H_1(\textbf{A}_s))+f_{\textbf{B}}(\textbf{r}_e)=\textbf{t}^*$ or $f_{\overline{\textbf{A}}'_{r^*}}(H_1(\textbf{A}'_s))+f_{\textbf{B}'}(\textbf{r}'_e)=\textbf{t}'^*$, where $\textbf{t}^*$ and $\textbf{t}'^*$ are defined as in \textbf{Game OW0} (we name this event by $\textsf{Event}_1$), then the challenger outputs $\bot$. 
				\iffalse
				Let $E_1$ is the event such a query happens. Then we have
							 \begin{align*}
							 \Pr[W_1]&=\Pr[W_1|E_1]\cdot \Pr[E_1]+\Pr[W_1|\neg E_1]\cdot \Pr[\neg E_1]\\
							 &=\frac{1}{2}\cdot \Pr[E_1]+\Pr[W_0 \cap \neg E_1]\\
							 &=\frac{1}{2}\cdot \Pr[E_1]+\Pr[W_0 ]-\Pr[W_0 \cap  E_1]\\
							 &\geq -\frac{1}{2}\cdot \Pr[E_1]+\Pr[W_0 ]\\
							 \end{align*}
							 
							 \fi 
							 
				\item 			 Since the view of the adverary $\mathcal{A}_2$ is the same, except once the event $\textsf{Event}_1$ happens,  in both \textbf{Game OW2} and \textbf{Game OW1}, we have $$|\Pr[W_2]-\Pr[W_1]|\leq \epsilon_{f,CR}.$$

				\item[Game OW3.]  This game is same as \textbf{Game OW2}, except that  instead of  $\mathbf{B}, \mathbf{B}'$ being uniform in $ \mathbb{Z}_q ^{n \times m}$, the challenger uses $\textsf{GenTrap}(n, \overline{m},q, \sigma_1)$ to generate  $(\mathbf{B}, \mathbf{T}_{\textbf{B}}), (\mathbf{B}', \mathbf{T}_{\textbf{B}}') \in  \mathbb{Z}_q^{n \times m} \times \mathbb{Z}_q^{\overline{m} \times nk}$. 
				\item Due to the fact that $\mathbf{B}, \mathbf{B}'$ generated by \textsf{GenTrap} are close to uniform, then we have $\Pr[W_3]=\Pr[W_2].$

	  \item[Game OW4.] This game is same as \textbf{Game OW3}, except that in the \textbf{Setup} phase, for the target receiver $r^*$, the challenger generates as follows:
					  
			\begin{enumerate}
					  			 
		 \item Choose $\textbf{t}^*, \textbf{t}'^*\in \mathbb{Z}_q^n$ uniformly at random. The challenger uses  $\textbf{t}^*, \textbf{t}'^*$  to build $\mathbf{A}_{r^*}, \mathbf{A}'_{r^*}$.
		\item  Choose $ \mathbf{T}_{r^*}, \mathbf{T}'_{r^*} \leftarrow D_{\mathbb{Z}^{\overline{m} \times n k }, \sigma_1} $ and then set $\mathbf{A}_{r^*}=[\overline{\mathbf{A}}_{r^*}|-H_2(\textbf{t}^*)\mathbf{G}-\overline{\mathbf{A}}_{r^*}\cdot \mathbf{T}_{r^*}] \in \mathbb{Z}_q^{n \times m}$, $\mathbf{A}'_{r^*}=[\overline{\mathbf{A}}'_{r^*}|-H_2(\textbf{t}'^*)\mathbf{G}-\overline{\mathbf{A}}'_{r^*}\cdot \mathbf{T}'_{r^*}] \in \mathbb{Z}_q^{n \times m}$.
		\item The public key for $r^*$ is $pk_{r^*}=(\mathbf{A}_{r^*},\mathbf{A}'_{r^*})$ and the private key key for $r^*$ is $sk_{r^*}=(\mathbf{T}_{r^*},\mathbf{T}'_{r^*})$.
					  		 	  	\end{enumerate}

		\item In this game, once the adversary $\mathcal{A}_2$ makes a trapdoor query for $r^*$, the challenger still easily returns $\textbf{T}'_{r^*}$ to $\mathcal{A}_2$. In  \textbf{Game OW4}  the view of $\mathcal{A}_2$ is unchanged in comparison with in \textbf{Game OW3} since the distribution of  $\mathbf{A}_{r^*}, \mathbf{A}'_{r^*}$ is the same as that of $\mathbf{A}_{r}, \mathbf{A}'_{r}$ for all $r \neq r^*$ which are generated in Step 1 of the \textsf{KG} algorithm.  Therefore,
		$$\Pr[W_4]=\Pr[W_3].$$

				\item[Game OW5.] This game is same as \textbf{Game OW4}, except that in the \textbf{Challenge} phase, on  the challenge message $\mu^*$,  the challenger does the following:
				\begin{enumerate}
				%	\item Choose randomly $b \xleftarrow{\$} \{0,1\}$.
				\item Compute $\textbf{r}_e^* \leftarrow \textsf{SampleD}(\textbf{T}_\textbf{B}, \textbf{B}, (\textbf{t}^*-f_{\overline{\textbf{A}}_{r^*}}(H_1(\textbf{A}_{s^*}))),\alpha q)$,  ${\textbf{r}'_e}^* \leftarrow \textsf{SampleD}(\textbf{T}'_\textbf{B}, \textbf{B}', (\textbf{t}'^*-f_{\overline{\textbf{A}}'_{r^*}}(H_1(\textbf{A}'_{s^*}))),\alpha q)$.  
					\item Sample  $\textbf{s}, \textbf{s}' \xleftarrow{\$} \mathbb{Z}_q^{n}$, \quad  $\hat{\textbf{x}}_0, \hat{\textbf{x}}_0'\leftarrow D_{\mathbb{Z}^{\overline{m}}, \alpha q}$, \quad  $\textbf{x}_1, \textbf{x}_1'\leftarrow D_{\mathbb{Z}^{\ell}, \alpha q}$.
					\item Compute $\hat{\textbf{c}}_0=\textbf{s}^t\overline{\mathbf{A}}_{r^*}+ \hat{\textbf{x}}^t_0 \in \mathbb{Z}_q^{\overline{m}}$, $\overline{\textbf{c}}_1=\textbf{s}^t\textbf{U}+ \textbf{x}^t_1 \in \mathbb{Z}_q^{\ell}$,
					$\hat{\textbf{c}}'_0=(\textbf{s}')^t\overline{\mathbf{A}}_{r^*}'+ (\hat{\textbf{x}}_0')^t \in \mathbb{Z}_q^{\overline{m}}$, $\overline{\textbf{c}}'_1=(\textbf{s}')^t\textbf{U}'+ (\textbf{x}'_1)^t \in \mathbb{Z}_q^{\ell}$,
					%	\item Compute ${\textbf{r}_e}^* \leftarrow \textsf{SampleD}(\textbf{T}_\textbf{B}, \textbf{B}, (\textbf{t}^*-f_{\overline{\textbf{A}}_{r^*}}(\textbf{A}_{s^*})))$,\\  ${\textbf{r}'_e}^* \leftarrow \textsf{SampleD}(\textbf{T}'_\textbf{B}, \textbf{B}', (\textbf{t}'^*-f_{\overline{\textbf{A}}'_{r^*}}(\textbf{A}'_{s^*})))$  
					%	\item $\mathbf{A}_{r^*,\textbf{t}^*}=\mathbf{A}_{r^*}+[\textbf{0}|H_2(\textbf{t}^*)\mathbf{G}]=[\overline{\mathbf{A}}_{r^*}|-\overline{\mathbf{A}}_{r^*}\cdot \mathbf{T}_{r^*}] \in \mathbb{Z}_q^{n \times m}$,\\ $\mathbf{A}'_{r^*,\textbf{t}'^*}=\mathbf{A}'_{r^*}+[\textbf{0}|H_2(\textbf{t}'^*)\mathbf{G}]=[\overline{\mathbf{A}}_{r^*}|-\overline{\mathbf{A}}'_{r^*}\cdot \mathbf{T}'_{r^*}] \in \mathbb{Z}_q^{n \times m}$.
			\item Set  ${(\textbf{c}^*_0)}^t:=(\hat{\textbf{c}}_0^t|\hat{\textbf{c}}_0^t)\textbf{T}_{r^*}$, \quad  ${(\textbf{c}'^*_0)}^t:=((\hat{\textbf{c}}'_0)^t|(\hat{\textbf{c}}'_0)^t)\textbf{T}'_{r^*}$.
		%	\item  Sample $\textbf{r}^*_s, \textbf{r}'^*_s \leftarrow D_{\mathbb{Z}^{m }, \alpha q}$, and compute  the challenge ciphertext ${ct}^*=({\textbf{c}_0}^*, {\textbf{c}_1}^*, {\textbf{r}_e}^*,{\textbf{c}'_0}^*, {\textbf{c}'_1}^*,{\textbf{r}'_e}^*, {\textbf{e}}^*)$ using Steps 6-8 of \textsf{SC}($pk_{r^*},sk_{s^*},\mu^*_b$) with $\overline{ct}^*=(\overline{\textbf{c}}^*_0, \overline{\textbf{c}}^*_1, \textbf{r}^*_e)$. 
		
			\item Sign on $\mu^*_b| pk_{r^*}|\overline{ct}^*$ with $\overline{ct}^*=(\textbf{c}^*_0, \overline{\textbf{c}}^*_1, \textbf{r}^*_e, \textbf{c}'^*_0, \overline{\textbf{c}}'^*_1, \textbf{r}'^*_e)$ to get the signature $(\textbf{e}^*, \textbf{r}_s^*)$ as usuall. 
			
			\iffalse  follows:					
							\begin{enumerate}
							
			\item Sample $\textbf{r}^*_s\leftarrow D_{\mathbb{Z}^{m }, \alpha q}$.	%Sample $\textbf{r} \leftarrow D_{\mathbb{Z}^{\overline{m} }, \alpha q}$ and let $(\textbf{r}^*_s)^t=(\textbf{r}^t|\textbf{r}^t\textbf{T}_{r^*})$.
			\item $\textbf{h}^*=(h^*_1, \cdots, h^*_n)=f_{\overline{\textbf{A}}_{s^*}}(H_3(\mu^*_b| pk_{r^*}|\overline{ct}^*))+f_{\textbf{B}}(\textbf{r}^*_{s})\in \mathbb{Z}_q^{n}$.
			\item $\textbf{A}_{s^*,\textbf{h}^*}=[\mathbf{A}_{s^*}|\mathbf{C}_0+\sum_{i=1}^{n}h^*_i\cdot \mathbf{C}_i]\in \mathbb{Z}_q^{n \times (m+nk)}$.
			\item $\textbf{e}^* \in \mathbb{Z}^{m+nk} \leftarrow \textsf{SampleD}(\textbf{T}_{s^*}, \textbf{A}_{s^*,\textbf{h}^*},\textbf{u},\sigma_2)$.
									\end{enumerate}
									\fi	
			\item $\mu'\xleftarrow{\$} \mathcal{M}$, \quad  $\textbf{c}^*_1=\overline{\textbf{c}}^*_1+ \mu^* \cdot \lfloor q/2\rfloor$, \quad  $\textbf{c}'^*_1=\overline{\textbf{c}}'^*_1+ H(\mu') \cdot \lfloor q/2\rfloor$.	
			 
					\item Return $ct^*=(\textbf{c}^*_0, \textbf{c}^*_1, \textbf{r}^*_{e},\textbf{r}^*_{s},\textbf{c}'^*_0, \textbf{c}'^*_1,\textbf{r}'^*_e, \textbf{e}^*)$ to $\mathcal{A}_2$.
									\end{enumerate}

			\item 	 We have $\Pr[W_5]=\Pr[W_4]$ as the distributions of corresponding components in $ct^*$ in \textbf{Game OW5} and \textbf{Game OW4} are the same.
				\item[Game OW6.] This game is same as \textbf{Game OW5}, except that the challenge ciphertext $ct^*=(\textbf{c}_0^*, \textbf{c}_1^*, \textbf{r}_e^*,\textbf{r}_s^*, \textbf{c}'^*_0, \textbf{c}'^*_1,\textbf{r}'^*_e, \textbf{e}^*)$ is  chosen uniformly at random.  The advantage of $\mathcal{A}_2$ in this game is obviously zero, i.e.,  $\Pr[W_6]=0.$
				
			\item	At this point, we show that $|\Pr[W_6]-\Pr[W_5]| \leq \epsilon_{LWE}$ which is negligible by using  a reduction from the LWE assumption as in Theorem \ref{indtheorem}. \qed

			\end{description}
			%\textbf{Analysis.} Obviously, if $(\textbf{F},\textbf{c}^t)$ is the LWE instance then the view of $\mathcal{A}_2$ is as in \textbf{Game OW5}; while if $(\textbf{F},\textbf{c}^t)$ is uniform in $\mathbb{Z}_q^{n \times (2\overline{ m}+2\ell)} \times \mathbb{Z}_q^{2(\overline{m}+\ell)}$ then the view of $\mathcal{A}_2$ as in \textbf{Game OW6}. Therefore, if the views of $\mathcal{A}_2$ in \textbf{Game OW6} and \textbf{Game OW5} are different  then $\mathcal{B}_2$ can solve the decision LWE problem. \qed	
			
			%Indeed, we have the distribution of $[\mathbf{A}'_r||\mathbf{C}'_0]$, $[\mathbf{A}'_r||\mathbf{C}'_0]$ is negligibly close to uniformWe consider a sequence of games:

		\end{proof}

	Before stating the SUF-iCMA security, we recap  the so-called \textit{abort-resistant hash functions}, presented in \cite[Section 7.4.1]{ABB10}. We will exploit the hash functions in designing answers to the adversary's queries.
	
	\begin{definition}[{\cite[Definition 26]{ABB10}}] \label{def3}
		Let $\mathcal{H}:=\{H:X\rightarrow Y\}$ be a family of hash functions $H$ from $X$ to $Y$ where $0 \in Y$. For a set of $Q+1$ inputs $\overline{\textbf{h}}:=(\textbf{h}^*,\textbf{h}^{(1)}, \cdots, \textbf{h}^{(Q)})$, the non-abort probability of $\overline{\textbf{h}}$ is defined as
		$$\alpha(\overline{\textbf{h}}):=\Pr[H(\textbf{h}^*)=0 \text{ and }H(\textbf{h}^{(1)})\neq 0 \text{ and } \cdots \text{ and } H(\textbf{h}^{(Q)})\neq0],$$
		where the probability is over the random choice of $H$ in $\mathcal{H}$. And $\mathcal{H}$ is called $(Q, \alpha_{\min}, \alpha_{\max})$ abort-resistant if for all $\overline{\textbf{h}}:=(\textbf{h}^*,\textbf{h}^{(1)}, \cdots, \textbf{h}^{(Q)})$ and $\textbf{h}^* \notin \{\textbf{h}^{(1)}, \cdots, \textbf{h}^{(Q)} \}$, we have $ \alpha_{\min}\leq \alpha(\overline{\textbf{h}})\leq  \alpha_{\max}$.
	\end{definition}
	Particularly, we have the following result that will be applied to the security proof for the proposed signcryption construction. 
	
	\begin{lemma}[{\cite[Lemma 27]{ABB10}}] \label{lemma3}
		let $q$ be a prime and $0<Q<q$. Consider the family	$\mathcal{H}_{Wat}:=\{H_{\mathbf{x}}: \mathbb{Z}^{n}_q \setminus \{\mathbf{0}\}\rightarrow \mathbb{Z}_q:  \mathbf{x}=(x_1, \cdots, x_n) \in \mathbb{Z}^{n}_q \setminus \{\mathbf{0}\}\}$  defined  as
		$H_{\mathbf{x}}(\mathbf{h})=1+\sum_{i=1}^{n}x_ih_i \in \mathbb{Z}_q$   where  $ \mathbf{h}=(h_1, \cdots, h_n) \in \mathbb{Z}^{n}_q.$ Then $\mathcal{H}_{Wat}$ is $(Q, \frac{1}{q}(1-\frac{Q}{q}), \frac{1}{q})$ abort-resistant.
	\end{lemma}
	Now, it is the time we  state and prove the SUF-iCMA security for \textsf{SCET}.
	
	\begin{theorem}[SUF-iCMA] Our \textsf{SCET} is SUF-iCMA secure in the standard model provided that the SIS problem is intractable. In particular, assume that there is an adversarial algorithm $\mathcal{F}$ who can win the SUF-iCMA  game making at most $Q<q/2$ adaptive chosen message queries. Then, there is an algorithm $\mathcal{G}$ who is able to solve the $\mathsf{SIS}_{n, 2\overline{m}+2nk+1, q, \beta}$ problem, with $ \beta:=  2\sigma_1\cdot \sigma_2\cdot \sqrt{n+1} \cdot \frac{1}{\sqrt{2\pi}}\cdot (\sqrt{\overline{m}}+\sqrt{nk}) \cdot\sqrt{m+nk}  $..
	\end{theorem}
	\begin{proof} 
	 We assume by contradiction that if there exists a forger $\mathcal{F}$ who can break the SUF-iCMA security of the \textsf{SCET} scheme, then we can build from $\mathcal{F}$ an algorithm $\mathcal{G}$ that can find a solution to a given SIS problem. %
	 
	 We will give proof for the SUF-iCMA in two cases: \textbf{Case 1:} The forger  $\mathcal{F}$ forges on an unqueried message and \textbf{Case 2:} The  forger  $\mathcal{F}$ who forges on a queried message. 
	 Suppose that, $\mathcal{F}$ makes at most $Q$ adaptive signcryption queries on $(r^{(1)},s^{(1)},\mu^{(1)}), \cdots, (r^{(Q)},s^{(Q)},\mu^{(Q)}),$.
	 
	 We consider \textbf{Case 1} first.
	 
	   \textbf{SIS Instance.} Suppose that the algorithm  $\mathcal{G}$ is given the following SIS problem below:
		\begin{equation}\label{key}
		\bf{F}\cdot \bf{x}=\bf{0} \text{ (mod } q), \text{ where } \bf{F} \xleftarrow{\$} \mathbb{Z}_q^{n \times (2\overline{m}+2nk+1)}, \Vert \textbf{x}\Vert \leq \beta_1, 
		\end{equation}
		where $\beta_1= \sigma_2 \cdot \sqrt{n+1}\cdot (\sqrt{\overline{m}}+\sqrt{nk})\sqrt{m+nk+1} $.
	Then, $\mathcal{G}$ parses $\bf{F}$ as $\bf{F}:=[\overline{\bf{A}}|\bf{W}|\textbf{u}|\overline{\bf{A}}'|\bf{W'}]$, where $\overline{\bf{A}},\overline{\bf{A}}' \in \mathbb{Z}_q^{n \times \overline{m}},$ $ \mathbf{u} \in \mathbb{Z}_q ^{n}$, and $ \mathbf{W}, \bf{W}'\in \mathbb{Z}_q^{n \times nk}$. 
		The algorithms  $\mathcal{G}$ and $\mathcal{F}$ play the following game:
		
	 \textbf{Setup.}   $\mathcal{G}$ simulates public parameters $pp$, public keys for $M$ senders and $N$ receivers as follows:
			\begin{itemize}
			\item Pick $n,q, k, m,\ell, n, \alpha,  N, M, \sigma_1, \sigma_2, H, H_1, H_2, H_3, \mathcal{M}$ as system parameters.
			
			\item Guess $s^* \xleftarrow{\$} \{1, \cdots, M\}$ to be the target sender that $\mathcal{F}$ would like to forge, and then set $\overline{\bf{A}}_{s^*}:=\overline{\bf{A}}$, $\overline{\bf{A}}'_{s^*}:=\overline{\bf{A}}'$,  $\mathbf{A}_{s^*}:=[\overline{\mathbf{A}}_{s^*}|\mathbf{W}] \in \mathbb{Z}_q^{n \times m}$, $\mathbf{A}'_{s^*}:=[\overline{\mathbf{A}}'_{s^*}|\mathbf{W}'] \in \mathbb{Z}_q^{n \times m}$. Recall  that $m=\overline{m}+nk.$
		%	\item Use $\textsf{TrapGen}(n, m,q)$ to generate  $(\mathbf{A}'_{s^*}, \mathbf{T}'_{s^*}) \in  \mathbb{Z}_q^{n \times m} \times \mathbb{Z}_q^{m \times m}$.
		
	%	\item Choose $ \mathbf{T}_{rs^*}, \mathbf{T}'_{s^*} \leftarrow D_{\mathbb{Z}^{\overline{m} \times n k }, \sigma} $ and then set $\mathbf{A}_{s^*}=[\overline{\mathbf{A}}_{s^*}|\mathbf{G}-\overline{\mathbf{A}}_{s^*}\cdot \mathbf{T}_{s^*}] \in \mathbb{Z}_q^{n \times m}$, $\mathbf{A}'_{s^*}=[\overline{\mathbf{A}}'_{r^*}|\mathbf{G}-\overline{\mathbf{A}}'_{s^*}\cdot \mathbf{T}'_{s^*}] \in \mathbb{Z}_q^{n \times m}$.
			\item Also, use $\textsf{GenTrap}(n, \overline{m},q, \sigma_1)$ to generate  $(\mathbf{B}, \mathbf{T}_{\textbf{B}}), (\mathbf{B}', \mathbf{T}_{\textbf{B}}') \in  \mathbb{Z}_q^{n \times m} \times \mathbb{Z}_q^{m \times m}$, and choose randomly $\mathbf{U},\mathbf{U}' \xleftarrow{\$} \mathbb{Z}_q ^{n \times \ell}$.
		%	\item Sample $2(n+1)$ matrices $\mathbf{R}_0, \cdots,$ $ \mathbf{R}_{n}$, $\mathbf{R}'_0, \cdots,$ $ \mathbf{R}'_{n} \in \cal{D}_{\mathbb{Z}^{m}, \eta}$, and choose randomly $\gamma_1, \cdots, \gamma_n$, $\gamma'_1, \cdots, \gamma'_n$ from $\mathbb{Z}_q$, while fix $\gamma_0:=1, \gamma'_0:=1 \in \mathbb{Z}_q$.
			\item For $i \in \{0, \cdots, n\}$, sample $ \mathbf{T}_{s^*,i}, \mathbf{T}'_{s^*,i} \leftarrow D_{\mathbb{Z}^{\overline{m}\times nk}, \sigma_1} $. Let $x_0:=1$. Choose $\textbf{x}:=(x_1, \cdots, x_n)\xleftarrow{\$}\mathbb{Z}_q^n\setminus \{\textbf{0}\}$ and for  $i \in \{0, \cdots, n\}$, let  $\textbf{H}_i=x_i\textbf{I}_n$ and set $\textbf{C}_i:=\textbf{H}_i\textbf{G}- \overline{\bf{A}}_{s^*} \textbf{T}_{s^*,i} \text{ (mod } q) \in \mathbb{Z}_q^{n \times nk}$, $\textbf{C}'_i:=\textbf{H}'_i\textbf{G}- \overline{\bf{A}}'_{s^*} \textbf{T}'_{s^*,i} \text{ (mod } q) \in \mathbb{Z}_q^{n \times nk}$. Obviously, by Lemma \ref{lemma3}, such an $\textbf{x}$ will define an abort-resistant hash function belonging to $\mathcal{H}_{Wat}$.
		%	$$\textbf{H}_i=\begin{cases} \textbf{0} \\ \textbf{I}_n \\ -6x + 3y + 2z \end{cases}.$$
	
		\item For each $j \in [Q]$: repeat choosing $\textbf{h}^{(j)}=(h^{(j)}_1, \cdots, h^{(j)}_{n})\in \mathbb{Z}_q^n$ uniformly at random until being such that $(1+\sum_{i=1}^{n}x_i \cdot h^{(j)}_i)\neq 0$.
			\item For all $s \in [M]\setminus \{s^*\}$ and all $r\in [N]$, use $\textsf{GenTrap}(n, \overline{m},q, \sigma_1)$ to generate  $(\mathbf{A}_r, \mathbf{T}_r)$, $ (\mathbf{A}'_r, \mathbf{T}'_r)$, $(\mathbf{A}_s, \mathbf{T}_s)$, $(\mathbf{A}'_s, \mathbf{T}'_s) \in  \mathbb{Z}_q^{n \times m} \times \mathbb{Z}_q^{\overline{m}\times nk}$.

			\item Set $pp=\{n,q,k, m,\ell, n, N, M, \alpha, \sigma_1, \sigma_2,  \mathcal{M}, (\bf{C}_i, \bf{C}'_i)_{i=0}^{n}, H, H_1, $ $H_2, $ $H_3\}$ as public parameters and  $pk_{r}=(\mathbf{A}_r, \mathbf{A}'_r)$, $pk_{s}=(\mathbf{A}_s, \mathbf{A}'_s)$ as public keys corresponding to each receiver $r \in [N]$, and each sender $s\in [M]$.
			\item Send $pp$, $pk_s$'s, $pk_r$'s all to the forger $\mathcal{F}$.
			\end{itemize}
		 \textbf{Queries.}  $\mathcal{F}$ can adaptively make  PKQ,  SCQ and  TGQ queries  polynomially many times  and in any order. Accordingly, $\mathcal{G}$ responds to the queries made by $\mathcal{F}$  as follows:
			\begin{itemize}
				\item For private key queries PKQ$(s)$: if  $s=s^*$, $\mathcal{G}$ rejects it. Otherwise, $\mathcal{G}$ returns the private key sk$_{s}=(\mathbf{T}_s,\mathbf{T}'_s )$ of a sender $\mathcal{S}_{s}$ to $\mathcal{F}$.
			
				\item For the $j$-th signcryption query SCQ$(r^{(j)},s^{(j)},\mu^{(j)})$ querie: If $s^{(j)}\neq s^*$, 	$\mathcal{G}$ sends the output $ct$ of  $\textsf{SC}(pk_{r^{(j)}}$, $sk_{s^{(j)}},\mu^{(j)})$ back to $\mathcal{F}$. Otherwise, $\mathcal{G}$ creates a ciphertext $ct$ on input $(pk_{r^{(j)}},sk_{s^*},\mu)$   as follows: 	
			\begin{enumerate}
			\item  $\textbf{r}_e,  \textbf{r}'_e \leftarrow D_{\mathbb{Z}^{m }, \alpha q}$, $\textbf{t}=f_{\overline{\textbf{A}}_r}(H_1(\textbf{A}_{s^*}))+f_{\textbf{B}}(\textbf{r}_e)\in \mathbb{Z}_q^{n}$,\\
			$\textbf{t}'=f_{\overline{\textbf{A}}'_r}(H_1(\textbf{A}'_{s^*}))+f_{\textbf{B}'}(\textbf{r}'_e)\in \mathbb{Z}_q^{n}$.
			\item $\textbf{A}_{r,\textbf{t}}=\textbf{A}_{r}+[\textbf{0}|H_2(\textbf{t})\mathbf{G}]\in \mathbb{Z}_q^{n \times m}=[\overline{\mathbf{A}}_r|H_2(\textbf{t})\mathbf{G}-\overline{\mathbf{A}}_r\cdot \mathbf{T}_r]\in \mathbb{Z}_q^{n \times m}$,\\
			$\textbf{A}'_{r,\textbf{t}}=\textbf{A}'_{r}+[\textbf{0}|H_2(\textbf{t}')\mathbf{G}]\in \mathbb{Z}_q^{n \times m}=[\overline{\mathbf{A}}'_r|H_2(\textbf{t}')\mathbf{G}-\overline{\mathbf{A}}'_r\cdot \mathbf{T}'_r]\in \mathbb{Z}_q^{n \times m}$.
									
			\item $\textbf{s}, \textbf{s}' \xleftarrow{\$} \mathbb{Z}_q^{n}$, \quad $\textbf{x}_0, \textbf{x}'_0\leftarrow D_{\mathbb{Z}^{m}, \alpha q}$, \quad  $\textbf{x}_1, \textbf{x}_1'\leftarrow D_{\mathbb{Z}^{\ell}, \alpha q}$.
			\item $\textbf{c}_0=\textbf{s}^t\textbf{A}_{r,\textbf{t}}+ \textbf{x}^t_0 \in \mathbb{Z}_q^{m}$, \quad \quad \quad  $\overline{\textbf{c}}_1=\textbf{s}^t\textbf{U}+ \textbf{x}^t_1 \in \mathbb{Z}_q^{\ell}$,\\
			$\textbf{c}'_0=(\textbf{s}')^t\textbf{A}'_{r,\textbf{t}}+ (\textbf{x}_0')^t \in \mathbb{Z}_q^{m}$, \quad  $\overline{\textbf{c}}'_1=(\textbf{s}')^t\textbf{U}'+ (\textbf{x}'_1)^t \in \mathbb{Z}_q^{\ell}$.
			\item $\overline{ct}=(\textbf{c}_0, \overline{\textbf{c}}_1, \textbf{r}_e, \textbf{c}'_0, \overline{\textbf{c}}'_1, \textbf{r}'_e)$.%, \overline{\textbf{c}}'_0, \overline{\textbf{c}}'_1, \textbf{r}'_e)$.

			\item Generate a signature on $\mu^{(j)}| pk_{r^{(j)}}|\overline{ct}$:
			\begin{enumerate}
		%	\item 	Choose $\textbf{h}=(h_1, \cdots, h_{n})\in \mathbb{Z}_q^n$ uniformly at random such that $(1+\sum_{i=1}^{n}x_i \cdot h_i)\neq 0$. 
		
	%	\item Sample $\textbf{r}'_{s}  \leftarrow D_{\mathbb{Z}^{m }, \alpha q}$.
			\item Compute ${\textbf{r}_s} \leftarrow \textsf{SampleD}(\textbf{T}_\textbf{B}, \textbf{B}, (\textbf{h}^{(j)}-f_{\overline{\textbf{A}}_{s^*}}(H_3(\mu^{(j)}| pk_r|\overline{ct})))),\alpha q)$.
		%	\item $\textbf{h}=(h_1, \cdots, h_n)=f_{\overline{\textbf{A}}_{s^*}}(H_3(\mu| pk_r|\overline{ct}))+f_{\textbf{B}}(\textbf{r}_{s})\in \mathbb{Z}_q^{n}$
			\item $\textbf{A}_{{s^*},\textbf{h}}=[\mathbf{A}_{s^*}|\mathbf{C}_0+\sum_{i=1}^{n}h_i^{(j)}\cdot \mathbf{C}_i]=[\mathbf{A}_{s^*}|\textbf{H} \textbf{G}-\overline{\textbf{A}}_{s^*}\textbf{T}_{s^*}]\in \mathbb{Z}_q^{n \times (m+nk)}$, with $\textbf{H}:=\mathbf{H}_0+\sum_{i=1}^{n}h_i\cdot \mathbf{H}_{i}=(1+\sum_{i=1}^{n}x_i \cdot h_i^{(j)})\textbf{I}_n \neq 0$ and $\textbf{T}_{s^*}:=\mathbf{T}_{s^*,0}+\sum_{i=1}^{n}h_i\cdot \mathbf{T}_{s^*,i}$. Note that,	by Lemma \ref{lemma4},  $\textbf{T}_{s^*}$ is a $\textbf{G}$-trapdoor for $\textbf{A}_{{s^*},\textbf{h}}$ with tag $\textbf{H}$. %Indeed,	\textcolor{red}{by definition $\textbf{H}_i$'s are scalar matrices then we have $\textbf{H}$ is invertible.}
			\item $\textbf{e} \in \mathbb{Z}^{ m+nk} \leftarrow \textsf{SampleD}(\textbf{T}_{s^*}, \textbf{A}_{{s^*},\textbf{h}},\textbf{u},\sigma_2)$.
			\item $(\textbf{e},\textbf{r}_s)$ is the signature.
		\end{enumerate}
			\item  $\textbf{c}_1=\overline{\textbf{c}}_1+ \mu^{(j)} \cdot \lfloor q/2\rfloor$, \quad $\textbf{c}'_1=\overline{\textbf{c}}'_1+ H(\mu^{(j)}) \cdot \lfloor q/2\rfloor$.
		\item Output $ct=(\textbf{c}_0, \textbf{c}_1, \textbf{r}_e, \textbf{r}_s,\textbf{c}'_0, \textbf{c}'_1,\textbf{r}'_e, \textbf{e})$.
		\end{enumerate}
							
		\item For unsigncryption query USQ$(r,s,ct)$: $\mathcal{G}$ simply sends the output $\mu$/$\bot$ of $\textsf{USC}(sk_{r},pk_{s}$, $ct)$ back to $\mathcal{F}$ since $\mathcal{G}$ has $(\bf{T}_r,\bf{T}'_r)$ for all $r \in [N]$.
					
		\item For tag query TGQ$(r)$: $\mathcal{G}$ simply sends the output $tg_{r}:=\bf{T}'_r$ back to $\mathcal{F}$.
			\end{itemize}
			
		 \textbf{Forge.} The forger $\mathcal{F}$  outputs an index  $r^* $ of some receiver and its corresponding key-pair $(pk_{r^*}, sk_{r^*})$, and a new valid ciphertext $ct^*=(\textbf{c}^*_0, \textbf{c}^*_1, \textbf{r}^*_e,\textbf{c}'^*_0,$ $ \textbf{c}'^*_1,\textbf{r}'^*_e, \textbf{e}^*)$ on message $\mu^*$, i.e., $ct^*$  must be not the output of any previuos query SCQ$(r,s,\mu)$, and $\mathsf{USC}(sk_{r^*},pk_{s^*},ct^*) \neq \bot.$ 
			
		\textbf{Analysis.} At the moment, $\mathcal{G}$ proceeds the following steps: 
			\begin{itemize}
				\item Compute $\textbf{h}^*=(h^*_1, \cdots, h^*_n)\gets f_{\overline{\textbf{A}}_{s^*}}(H_3(\mu^*| pk_{r^*}|\overline{ct}^*))+f_{\textbf{B}}(\textbf{r}_{s^*})\in \mathbb{Z}_q^{n}$ and check if $\textbf{H}:=\mathbf{H}_0+\sum_{i=1}^{n}h^*_i\cdot \mathbf{H}_{i}=(1+\sum_{i=1}^{n}x_i \cdot h_i^*)\textbf{I}_n=0$. If not, $\mathcal{G}$ aborts the game. Otherwise, it computes $\textbf{A}_{{s^*},\textbf{h}^*}=[\mathbf{A}_{s^*}|\mathbf{C}_0+\sum_{i=1}^{n}h^*_i\cdot \mathbf{C}_i]=[\overline{\textbf{A}}_{s^*}|\textbf{W}|-\overline{\textbf{A}}_{s^*}\textbf{T}_{s^*}]\in \mathbb{Z}_q^{n \times (m+nk)}$. Note that, the probability that $\mathcal{G}$ aborts the game is negligible by appropriately choosing parameters via Lemma \ref{lemma3}.
				\item From  $\textbf{A}_{s^*,\textbf{h}^*}\cdot \textbf{e}^*=\textbf{u} \!\! \mod q$ and $\|\textbf{e}^*\|\leq \sigma_2\sqrt{m+nk}$, we have $[\overline{\textbf{A}}_{s^*}|\textbf{W}|-\overline{\textbf{A}}_{s^*}\textbf{T}_{s^*}]\cdot \bf{e}^*=\bf{u} \text{ (mod } q),$ equivalently, 
	
	$$[\overline{\textbf{A}}_{s^*}|\textbf{W}] \cdot \left(\begin{bmatrix}
	\mathbf{I}_{\overline{m}} & \textbf{0} &-\mathbf{T}_{s^*} \\
	\textbf{0}&\textbf{I}_{nk}& \textbf{0}\\
	\end{bmatrix} \cdot \bf{e}^*\right)=\bf{u} \text{ (mod } q).$$

\item 
		Let $\widehat{\mathbf{x}}:=\begin{bmatrix}
			\mathbf{I}_{\overline{m}} & \textbf{0} &-\mathbf{T}_{s^*} \\
			\textbf{0}&\textbf{I}_{nk}& \textbf{0}\\
			\end{bmatrix} \cdot \bf{e}^* \neq \textbf{0}$,  then $[\overline{\textbf{A}}|\textbf{W}|\textbf{u}]\cdot  \bigl( \begin{smallmatrix}
				   \widehat{\mathbf{x}}\\ -1
				 \end{smallmatrix} \bigr)=\bf{0}\text{ (mod } q)$.
	Hence, $\mathcal{G}$ gets a solution $\mathbf{x}$ to the SIS problem \eqref{key}, i.e.,  $\mathbf{F}\cdot \mathbf{x}=\bf{0}\text{ (mod } q),$ with $\mathbf{x}= \bigl( \begin{smallmatrix}
	   \widehat{\mathbf{x}}\\ -1\\\mathbf{0}
	 \end{smallmatrix} \bigr)$, and by Lemma \ref{bound},
	\begin{align*}
	\| \mathbf{x}\|=\| \widehat{\mathbf{x}}\| +1	&\leq  s_1(\mathbf{T}_{s^*}) \cdot \| \mathbf{e}^*\|+1\\
	&	\leq \sigma_1\cdot \sigma_2\cdot \sqrt{n+1} \cdot \frac{1}{\sqrt{2\pi}}\cdot (\sqrt{\overline{m}}+\sqrt{nk}) \cdot\sqrt{m+nk}+1.
	\end{align*}

	\end{itemize}

Now, we consider \textbf{Case 2}.  Suppose that the algorithm  $\mathcal{G}$ is given the following SIS problem below:
			\begin{equation}\label{key2}
			\bf{F}\cdot \bf{x}=\bf{0} \text{ (mod } q), \text{ where } \bf{F} \xleftarrow{\$} \mathbb{Z}_q^{n \times (2\overline{m}+2nk)}, \Vert \textbf{x}\Vert \leq \beta_2,
			\end{equation}
	where $ \beta_2:=  2\sigma_1\cdot \sigma_2\cdot \sqrt{n+1} \cdot \frac{1}{\sqrt{2\pi}}\cdot (\sqrt{\overline{m}}+\sqrt{nk}) \cdot\sqrt{m+nk}  $. 	Then, $\mathcal{G}$ parses $\bf{F}$ as $\bf{F}:=[\overline{\bf{A}}|\bf{W}|\overline{\bf{A}}'|\bf{W'}]$, where $\overline{\bf{A}},\overline{\bf{A}}' \in \mathbb{Z}_q^{n \times \overline{m}},$  and $ \mathbf{W}, \bf{W}'\in \mathbb{Z}_q^{n \times nk}$. 
			The algorithms  $\mathcal{G}$ and $\mathcal{F}$ play the following game:
			
		 \textbf{Setup.}   This phase is the same as  the \textbf{Setup} phase of \textbf{Case 1}, except that $\mathcal{G}$ does the following:
				\begin{itemize}
			%	\item Pick $n,q, k, m,\ell,  n, N, M, \sigma_1, \alpha q, \sigma_2, H, H_1, H_2, H_3, \mathcal{M}$ as system parameters.
				
				\item Randomly guess $s^* \xleftarrow{\$} \{1, \cdots, M\}$ and $r^* \xleftarrow{\$} \{1, \cdots, N\}$ to be the target sender and the target receiver respectively, that $\mathcal{F}$ would like to forge, and then set $\overline{\bf{A}}_{s^*}:=\overline{\bf{A}}$, $\overline{\bf{A}}'_{s^*}:=\overline{\bf{A}}'$,  $\mathbf{A}_{s^*}:=[\overline{\mathbf{A}}_{s^*}|\mathbf{W}] \in \mathbb{Z}_q^{n \times m}$, $\mathbf{A}'_{s^*}:=[\overline{\mathbf{A}}'_{s^*}|\mathbf{W}'] \in \mathbb{Z}_q^{n \times m}$. Recall  that $m=\overline{m}+nk.$

				\item Repeat choosing $\textbf{h}^{(j)}=(h^{(j)}_1, \cdots, h^{(j)}_{n})\in \mathbb{Z}_q^n$ uniformly at random until   $(1+\sum_{i=1}^{n}x_i \cdot h^{(j_0)}_i)= 0$ for some $ j_0 \in [Q]$ and  $(1+\sum_{i=1}^{n}x_i \cdot h^{(j)}_i)\neq 0$ for each $j \in [Q]\setminus \{j_0\}$. Let $\textbf{h}^{(j_0)}$ to be the one corresponding to the target query $(r^*,s^*, \mu^*)$ and let $\textbf{h}^*=(h_1^*,\cdots, h_n^*) \gets \textbf{h}^{(j_0)}$.
		%			\item For all $s \in [M]\setminus \{s^*\}$ and all $r\in [N]$, use $\textsf{GenTrap}(n, m,q)$ to generate  $(\mathbf{A}_r, \mathbf{T}_r)$, $ (\mathbf{A}'_r, \mathbf{T}'_r)$ (with trapdoor tag $\textbf{H}=\textbf{0}$) and $(\mathbf{A}_s, \mathbf{T}_s)$, $(\mathbf{A}'_s, \mathbf{T}'_s)$ (with trapdoor tag $\textbf{H}=\textbf{I}$).

	\item  Set $\textbf{A}_{{s^*},\textbf{h}^*}=[\mathbf{A}_{s^*}|\mathbf{C}_0+\sum_{i=1}^{n}h^*_i\cdot \mathbf{C}_i]=[\overline{\textbf{A}}_{s^*}|\textbf{W}|-\overline{\textbf{A}}_{s^*}\textbf{T}_{s^*}],$ where $\textbf{T}_{s^*}:=\mathbf{T}_{s^*,0}+\sum_{i=1}^{n}h^*_i\cdot \mathbf{T}_{s^*,i}$.
			\item Choose ${\textbf{e}^*}^{(1)}  \leftarrow D_{\mathbb{Z}^{m+nk},\sigma_2}$, set $ \mathbf{u}:=\textbf{A}_{s^*,\textbf{h}^*}\cdot {\textbf{e}^*}^{(1)} \in \mathbb{Z}_q ^{n}$ and send $\textbf{u}$ to $\mathcal{F}$ as a uniformly random one.

			%	\item Set $pp=\{n,q, k, m,\ell, n, N, M, \alpha, \sigma_1,  \sigma_2, H, H_1, H_2, H_3, \mathcal{M},  (\bf{C}_i, \bf{C}'_i)_{i=0}^{n}, $ $ \textbf{U}, \textbf{U}', \textbf{u}\}$ as public parameters and  $pk_{r}=(\mathbf{A}_r, \mathbf{A}'_r)$, $pk_{s}=(\mathbf{A}_s, \mathbf{A}'_s)$ as public keys corresponding to each receiver $r \in [N]$, and each sender $s\in [M]$.
		%		\item Send $pp$, $pk_s$'s, $pk_r$'s all to the forger $\mathcal{F}$.
				\end{itemize}
			 \textbf{Queries.}  For almost PKQ,  SCQ and  TGQ queries,$\mathcal{G}$ responds in the same way as in the \textbf{Queries} phase of \textbf{Case 1}, except that  with the target query SCQ$(r^*,s^*,\mu^*)$, $\mathcal{G}$ responds as follows: 
		\begin{enumerate}
	
\item Produce $\overline{ct}^*=(\textbf{c}^*_0, \overline{\textbf{c}}^*_1, \textbf{r}^*_e, \textbf{c}'^*_0, \overline{\textbf{c}}'^*_1, \textbf{r}'^*_e)$ as usual.%, \overline{\textbf{c}}'_0, \overline{\textbf{c}}'_1, \textbf{r}'_e)$.
	\item Generate a signature on $\mu^*| pk_{r^*}|\overline{ct}^*$:
		\begin{enumerate}

			\item $\textbf{r}^*_s \leftarrow \textsf{SampleD}(\textbf{T}_{\textbf{B}}, \textbf{B},\textbf{h}^*-f_{\overline{\textbf{A}}_{s^*}}(H_3(\mu^*| pk_{r^*}|\overline{ct}^*)),\alpha q)$.
	
				\item Set $({\textbf{e}^*}^{(1)},{\textbf{r}_s^*})$ to be the signature.
										\end{enumerate}
		\item $\textbf{c}^*_1=\overline{\textbf{c}}^*_1+ \mu \cdot \lfloor q/2\rfloor$, \quad $\textbf{c}'^*_1=\overline{\textbf{c}}'^*_1+ H(\mu) \cdot \lfloor q/2\rfloor$.
		\item Return ${ct^*}^{(1)}=({\textbf{c}^*_0}, {\textbf{c}^*_1}, {\textbf{r}^*_e},  {\textbf{r}^*_s}, {\textbf{c}'^*_0}, {\textbf{c}'^*_1},{\textbf{r}'^*_e}, {\textbf{e}^*}^{(1)})$ to $\mathcal{F}$.
									\end{enumerate}

			 \textbf{Forge.} The forger $\mathcal{F}$  outputs the target receiver's key-pair $(pk_{r^*}, sk_{r^*})$ together with a new valid ciphertext ${ct^*}^{(2)}=({\textbf{c}^*_0}, {\textbf{c}^*_1}, {\textbf{r}^*_e},  {\textbf{r}^*_s}, {\textbf{c}'^*_0}, {\textbf{c}'^*_1},{\textbf{r}'^*_e}, {\textbf{e}^*}^{(2)})$ on the message $\mu^*$. Note that, ${\textbf{e}^*}^{(2)}\neq { \textbf{e}^*}^{(1)}$ while ${\textbf{c}^*_0}, {\textbf{c}^*_1}, {\textbf{r}^*_e},  {\textbf{r}^*_s}, {\textbf{c}'^*_0}, {\textbf{c}'^*_1},{\textbf{r}'^*_e}$ are unchanged due to the validity of ${ct^*}^{(2)}$ corresponding to $\textbf{A}_{s^*,\textbf{h}^*}$.
				
			\textbf{Analysis.} At the moment, $\mathcal{G}$ proceeds the following steps: 
				\begin{itemize}
				
					\item Recall $\textbf{A}_{{s^*},\textbf{h}^*}=[\mathbf{A}_{s^*}|\mathbf{C}_0+\sum_{i=1}^{n}h^*_i\cdot \mathbf{C}_i]=[\overline{\textbf{A}}_{s^*}|\textbf{W}|-\overline{\textbf{A}}_{s^*}\textbf{T}_{s^*}]$.
					\item From  $\textbf{A}_{{s^*},\textbf{h}^*}\cdot{\textbf{e}^*}^{(1)}=\textbf{u}\!\! \mod q$ and $\textbf{A}_{s^*,\textbf{h}^*}\cdot {\textbf{e}^*}^{(2)}=\textbf{u} \!\! \mod q$, we have
		$$[\overline{\textbf{A}}_{s^*}|\textbf{W}] \cdot \left(\begin{bmatrix}
		\mathbf{I}_{\overline{m}} & \textbf{0} &-\mathbf{T}_{s^*} \\
		\textbf{0}&\textbf{I}_{nk}& \textbf{0}\\
		\end{bmatrix} \cdot ({\textbf{e}^*}^{(1)}-{\textbf{e}^*}^{(2)})\right)=\bf{0} \text{ (mod } q).$$

	\item 
			Let $\widehat{\mathbf{x}}:=\begin{bmatrix}
				\mathbf{I}_{\overline{m}} & \textbf{0} &-\mathbf{T}_{s^*} \\
				\textbf{0}&\textbf{I}_{nk}& \textbf{0}\\
				\end{bmatrix} \cdot ({\textbf{e}^*}^{(1)}-{\textbf{e}^*}^{(2)})$,  then $[\overline{\textbf{A}}|\textbf{W}]\cdot  
					   \widehat{\mathbf{x}}
					 =\bf{0}\text{ (mod } q),$.
		Hence, $\mathcal{G}$ gets a solution $\mathbf{x}$ to the SIS problem \eqref{key2}, i.e.,  $\mathbf{F}\cdot \mathbf{x}=\bf{0}\text{ (mod } q),$ with $\mathbf{x}= \bigl( \begin{smallmatrix}
		   \widehat{\mathbf{x}}\\\mathbf{0}
		 \end{smallmatrix} \bigr)$, and 	 
	\begin{align*}
		\| \mathbf{x}\|=\| \widehat{\mathbf{x}}\| 	&\leq  s_1(\mathbf{T}_{s^*}) \cdot \| {\textbf{e}^*}^{(1)}-{\textbf{e}^*}^{(2)}\|\\
		&	\leq 2\sigma_1\cdot \sigma_2\cdot \sqrt{n+1} \cdot \frac{1}{\sqrt{2\pi}}\cdot (\sqrt{\overline{m}}+\sqrt{nk}) \cdot\sqrt{m+nk},
	\end{align*}
		 by Lemma \ref{bound}.

		\end{itemize}
			It remains  to prove that $\widehat{\mathbf{x}}\neq 0$ with overwhelming probability. Let $\textbf{w}:={\textbf{e}^*}^{(1)}-{\textbf{e}^*}^{(2)} \neq \textbf{0}$ and parse $\mathbf{w}= \bigl( \begin{smallmatrix}
						  \mathbf{w}_1\\\mathbf{w}_2\\\textbf{w}_3
						 \end{smallmatrix} \bigr)$. Then  $\widehat{\mathbf{x}}=(\mathbf{w}_1-\textbf{T}_{s^*} \mathbf{w}_3, \mathbf{w}_2)$. Obviously, if $\mathbf{w}_2 \neq \textbf{0}$ or $\mathbf{w}_3:=(w_1, \cdots, w_{nk})=\textbf{0}$ then $\widehat{\mathbf{x}}\neq 0$. Otherwise, i.e., $\mathbf{w}_2 = \textbf{0}$ and $\mathbf{w}_3\neq \textbf{0}$,  we will show that $\overline{\textbf{A}}_{s^*}(\mathbf{w}_1-\textbf{T}_{s^*} \mathbf{w}_3)=\textbf{0}$ happens only with negligible probability. Indeed, without loss of generality, assume that  $w_{nk}\neq 0$ then $\overline{\textbf{A}}_{s^*}(\mathbf{w}_1-\textbf{T}_{s^*} \mathbf{w}_3)=\textbf{0}$ only if $\textbf{t}_{nk}\sim D_{\Lambda_q^{\bot}(\overline{\textbf{A}}_{{s^*}})+\textbf{c}, \sigma_1}$ and $\textbf{t}_{nk} \cdot w_{nk}=\textbf{y}$ for some $\textbf{y}\in \Lambda_q^{\bot}(\overline{\textbf{A}}_{{s^*}})+\textbf{c}$ for any $\textbf{c}$ in span$(\Lambda_q^{\bot}(\overline{\textbf{A}}_{{s^*}}))$, where $\textbf{t}_{nk}$ is the $nk$-th column of $\textbf{T}_{s^*,0}$. Then by Lemma  \ref{min-entropy}, such a  $\textbf{t}_{nk}$ exists with negligible probability.	
						 
						 In conclusion, we choose the common SIS problem  $\mathsf{SIS}_{n, 2\overline{m}+2nk+1, q, \beta}$  for both two cases  (i.e., \textbf{Case 1} and \textbf{Case 2}) with $\beta:=\max \{\beta_1, \beta_2\}$ which should be  $ \beta:=  2\sigma_1\cdot \sigma_2\cdot \sqrt{n+1} \cdot \frac{1}{\sqrt{2\pi}}\cdot (\sqrt{\overline{m}}+\sqrt{nk}) \cdot\sqrt{m+nk}  $.
						 \qed
		\end{proof}

	\section{Parameter Selection} \label{para}
%	Basically, parameters are chosen to ensure the hardness of LWE and SIS problems on which our scheme is based as well as the trapdoor algorithms work correctly, and the security proofs work smoothly. 
	
	\begin{itemize}
		\item We take $n$ as the security parameter.
		\item For the gadget-based trapdoor mechanism to work: Choose $ q\ge 2, \overline{m} \ge 1 $, $ k = \lceil \log_2  q \rceil $, and $m = O(n\log q)$.
		
		\item Also choose $q, n, \overline{m} ,\ell, Q$ such  that  $H$ is one-way, $H_1$ is collision-resistant,    $H_3$ is  universal, $\mathcal{H}_{Wat}$ (in Lemma \ref{lemma3}) is a family of  abort-resistant hash functions, and  the functions $ f_{\overline{\textbf{A}}_{r}}(\cdot)+f_{\textbf{B}}(\cdot)$ are collision-resistance   for any $\overline{\mathbf{A}}_r \xleftarrow{\$} \mathbb{Z}_q^{n \times\overline{ m}}$ and any $\mathbf{B} \xleftarrow{\$} \mathbb{Z}_q^{n \times m}$.
		\item By Theorem  \ref{theo1}, in order for  the decisional-LWE $\mathsf{dLWE}_{n,2(\overline{m}+\ell), q,\alpha q}$ to be hard: $ q>2\sqrt{n}/\alpha $ for $\alpha=1/\textsf{poly}(n) \in (0,1)$.

		\item For the Gaussian parameter $\sigma_1$ used in \textsf{KeyGen} (which follows \textsf{GenTrap}):   $\sigma_1 \geq \omega(\sqrt{\log n})$. Note also that, all private keys for senders and users $\mathbf{T}_s, \mathbf{T}_r$ satisfy that $ s_1(\mathbf{T}_s), s_1(\mathbf{T}_r) \leq \sigma_1 \cdot \frac{1}{\sqrt{2\pi}}\cdot (\sqrt{\overline{m}}+\sqrt{nk})$ by Lemma \ref{supnorm}.
		
		\item For  the parameter $\alpha$ used in \textsf{Invert}:  We should choose $\alpha$ such that  $1/\alpha \geq 2 \sqrt{5 (s_1(\mathbf{R})^2+1)}\cdot \omega(\sqrt{\log n})$ by Lemma \ref{trapdoor}.
		\item For the Gaussian parameter $\sigma_2$ used in \textsf{SampleD}:  $ \sigma_2 \geq \sqrt{7( s_1(\mathbf{T}_s)^2+1)}\cdot \omega(\sqrt{\log n})$ (see \cite[Section 5.4]{MP12}).

		\item For the 
		$\mathsf{SIS}_{n, 2\overline{m}+2nk+1, q, \beta}$ problem has a solution and to be hard  (in the SUF-iCMA security proof):
		$\beta \geq \sqrt{2\overline{m}+2nk+1}\cdot q^{n/(2\overline{m}+2nk+1)}, $ $ q \geq \beta\cdot \omega(\sqrt{n\log n}),$ $ \beta:=  2\sigma_1\cdot \sigma_2\cdot \sqrt{n+1} \cdot \frac{1}{\sqrt{2\pi}}\cdot (\sqrt{\overline{m}}+\sqrt{nk}) \cdot\sqrt{m+nk} .$
	\end{itemize}
	%$q=\textsf{ploy}(n)$, $m=\lceil 6n\log q\rceil$, $l=\lceil n\log q\rceil$

	\section{Insecurity  of the Signcryption by Lu et al. \cite{LWJ+14}} \label{attack}
	
	\subsection{Description}	
	
	The signcryption  by Lu et al. \cite{LWJ+14} (called \textsf{Lu-SC}) exploits the basis-based trapdoor mechanism by \cite{GPV08a} which consists of algorithms \textsf{TrapGen}, \textsf{ExtBasis}, \textsf{SamplePre}.  The \textsf{Lu-SC} includes  the following algorithms:
	
	\begin{itemize}
		\item \textsf{Setup($1^{n}$)}: On input the security parameter $n$, performs the following:
				\begin{enumerate}
					\item Generate common parameters: $q=\textsf{poly}(n)$, $m=\lceil 6n\log q\rceil$, $\tilde{L}=O(\sqrt{n \log q})$, $\sigma\geq \tilde{L} \omega(\sqrt{\log m})$, error rate $\alpha=1/\textsf{poly}(n)$ such that $\alpha q >2\sqrt{n}$.
				\item Sample randomly and independently  matrices $\mathbf{C}_0, \cdots, \mathbf{C}_{\tau} \in \mathbb{Z}_q ^{n \times m}$.
			\item A collision-resistant hash function $H_1:\{0,1\}^{*} \rightarrow \{0,1\} ^{\tau}$, a universal hash function $H_2:\{0,1\}^{*} \rightarrow  \mathbb{Z}_q^n $.
			\item A message space $\mathcal{M}=  \mathbb{Z}_q^n$.
				\end{enumerate}	
		\item \textsf{KeyGen($n$)}: Do the following:
				\begin{enumerate}
					\item Use $\textsf{TrapGen}(1^n)$ to generate the matrix-pairs $(\mathbf{A}_R, \mathbf{T}_R)$,  $(\mathbf{A}_S, \mathbf{T}_S)$, where each pair belongs to $ \mathbb{Z}_q^{n \times m} \times \mathbb{Z}_q^{m \times m}$.
				\item  Sample randomly matrices $\mathbf{B}_R,  \mathbf{B}_S$, each from $\mathbb{Z}_q^{n \times m}$.
				\item Return
				\begin{itemize}
					\item $\mathsf{pk}_R:=(\mathbf{A}_R, \mathbf{B}_R)$, and  $\mathsf{sk}_R:=\textbf{T}_{R}$ as public key and secret key for a receiver $\mathcal{R}$,
					\item $\mathsf{pk}_S:=(\mathbf{A}_S, \mathbf{B}_S)$, and $\mathsf{sk}_S:=\textbf{T}_{S}$ as public key and secret key for a sender $\mathcal{S}$.
				\end{itemize}
				\end{enumerate}
	
	\item \textsf{SignCrypt}($ \mu, \mathsf{sk}_S, \mathsf{pk}_R$): On input  a plaintext $mu \in \mathcal{M}$,  $\mathsf{pk}_R:=(\mathbf{A}_R, \mathbf{B}_R)$, $\mathsf{sk}_S:=\mathbf{T}_S$, do the following:

				\begin{enumerate}
					\item   $\mathbf{h}=(h_i)_{i\in[\tau]}=H_	1(\mu,\mathsf{pk}_R) \in \mathbb{Z}_q^{\tau}$.
				\item $\mathbf{F}_{\mu}:=[\mathbf{A}_S|| \mathbf{C}_0+\sum_{i=1}^{\tau} (-1)^{h_i}\mathbf{C}_i] \in \mathbb{Z}_q^{n \times 2m}$.
				\item Compute $\mathbf{T}_{\mu}$ as the short basis for $\Lambda_q^{\bot}(\mathbf{F}_{\mu})$ from $\mathbf{T}_{S}$ using \textsf{ExtBasis}; $\mathbf{v}\leftarrow \textsf{SamplePre}(\mathbf{F}_{\mu}, \mathbf{T}_{\mu}, \mathbf{0},\sigma)$, $\mathbf{v}\in \mathbb{Z}_q^{2m}$, i.e., $\mathbf{F}_{\mu}\cdot \mathbf{v}=\mathbf{0} \textbf{ (mod } q)$.
				\item
				$\mathbf{t}=H_2(\mu, \textsf{pk}_S, \textsf{pk}_R, \mathbf{v})$ and $\mathbf{c}:=\mathbf{t}+\mu \bmod q$. \footnote{Actually, in \cite{LWJ+14}, the authors wrote $\mathbf{c}:=\mathbf{t}\oplus \mu$. However, since $\mu$ and $\mathbf{t}$ belong to $\mathbb{Z}_q^n$, we think that it should be $\mathbf{c}:=\mathbf{t} + \mu \bmod q$. }.
				
				\item $\mathbf{e} \leftarrow \widetilde{\Psi}_{\alpha}^{2m}$, $\mathbf{b}_1=[\mathbf{A}_R||\mathbf{C}_0]^{T}\cdot \mathbf{t}+\mathbf{v}\in \mathbb{Z}_q^{ 2m}$, $\mathbf{b}_2=[\mathbf{B}_R||\mathbf{C}_1]^{T}\cdot \mathbf{t}+\mathbf{e} \in \mathbb{Z}_q^{2m}$.

				\item Output ciphertext $\mathsf{CT}=(\mathbf{c}, \mathbf{b}_1, \mathbf{b}_2)$.
				\end{enumerate}

	\item \textsf{UnSignCrypt}($ \textsf{CT}, \mathsf{pk}_S,\mathsf{sk}_R$): On input $ \mathsf{pk}_S:=(\mathbf{A}_R, \mathbf{B}_R)$, $\mathsf{sk}_R:=(\textbf{T}_{R})$, $\textsf{CT}:=(\mathbf{c}, \mathbf{b}_1, \mathbf{b}_2)$, do the following:
			\begin{enumerate}
	
				\item Compute $\mathbf{t}$ and $\mathbf{v}$ from $\mathbf{b}_1$ using \textsf{ExtBasis} and \textsf{Invert} with the help of $\mathbf{T}_R$
				\item Compute $\mathbf{e}=\mathbf{b}_2-[\mathbf{B}_R||\mathbf{C}_1]^{T}\cdot \mathbf{t}$ and check whether $0<\Vert \mathbf{e}\Vert \leq \sigma\sqrt{2m}$ or not. If not, reject it and return  $\bot$
				\item Compute $\mu=\mathbf{c} - \mathbf{t} \bmod q$ and $\mathbf{h}=(h_i)_{i\in[\tau]} =H_1(\mu, \textsf{pk}_R)$
				\item Check whether $\mathbf{t} =H_2(\mu, \textsf{pk}_S, \textsf{pk}_R, \mathbf{v})$ or not. If not, reject it and output $\bot$
				 \item If {$\mathbf{v}\in \mathbb{Z}^{2m}$ and $0<\Vert \mathbf{v} \Vert \leq \sigma \sqrt{2m}$ and $[\mathbf{A}_S|| \mathbf{C}_0+\sum_{i=1}^{\tau} (-1)^{h_i}\mathbf{C}_i]\cdot \mathbf{v}= \mathbf{0} \text{ (mod } q)$}
				output $\mu$; otherwise, output $\bot$.
			\end{enumerate}

			\end{itemize}

\subsection{An Attack against IND-CPA}	
Recall that, in the challenge phase of the IND-CPA security, the adversary submits two plaintexts $\mu^*_0, \mu^*_1$ together with the target public  key $\textsf{pk}_{S^*}$. The challenger  then chooses uniformly at random a bit $b \in \{0,1\}$ and returns the challenge ciphertext $\textsf{CT}^* \leftarrow \textsf{SignCrypt}(\mu^*_b, \mathsf{sk}_{S^*}, \mathsf{pk}_{R^*})$ back to the adversary. The adversary wins the game if he can guess correctly the bit $b$.

\iffalse 
% For the \textsf{Lu-SC}, note that the adversary knows both receiver/sender public keys $\textsf{pk}_R$ and $\textsf{pk}_S$

%including $\textsf{pk}_{R^*} =(\textbf{A}_{R^*}, \textbf{ B}_{R^*})$ and $\textsf{pk}_{S^*} = (\textbf{A}_{S^*}, \textbf{B}_{S^*})$.
The adversary also knows  public parameters $\textbf{C}_0,...,\textbf{C}_\tau$  and $H_1,
H_2$.
\fi 

 Given the challenge ciphertext $\textsf{CT}^* = (\textbf{c}^*, \textbf{b}^*_1, \textbf{b}^*_2)$ of a plaintext either $\mu^*_0 $ or
$\mu^*_1$, the adversary is able to check the following conditions:

\begin{itemize}
	\item $\textbf{v }= \textbf{b}_1- [\textbf{A}_{R*} | \textbf{C}_0]^T \cdot (\textbf{c} - \mu^*_b)
$,  and 
$\textbf{e }= \textbf{b}_2 -[\textbf{B}_{R^*} | \textbf{C}_1]^T \cdot (\textbf{c} - \mu^*_b)
$
are small, 
\item
 $\textbf{c }- \mu^*_b = H_2(\mu^*_b, \textsf{pk}_{S^*}, \textsf{pk}_{R^*}, \textbf{v})$,
\item $ [\textbf{A}_{S^*} | \textbf{C}_0 + \sum^\tau_{i=1}(-1)^{h_i}\textbf{C}_i] \textbf{v} = \textbf{0 }\bmod q
,$ where $(h_i)_{i\in [\tau]}=H_1(\mu^*_b, \textsf{pk}_{R^*})$.
\end{itemize}
We see that the correct $b$ satisfies all above conditions, whilst the incorrect $b $ does not meet all the conditions with high probability. Then the adversary is able to
win the IND-CPA security game with high probability.

Moreover,  in the case that  both $b = 0$ and $b=1$ satisfy all the conditions, then the adversary is able to find a collision of
$H_2$.

	\section{Conclusions} \label{conc}
	In this work, we constructed a lattice-based signcryption scheme associated with a capacity of equality testing in the standard model. To the best of our knowledge, this scheme is the first post-quantum signcryption with equality test  in the literature. The proposed scheme satisfies confidentiality (IND-iCCA1 and OW-iCCA1) and the (strong) unforgeability under chosen message attack (SUF-iCMA) against insider attacks at the same time in which the former is based on the decisional-LWE assumption and the latter is guaranteed by the hardness of the SIS problem. We also showed that some lattice-based signcryptions in the literature   neither are secure nor  work correctly.  Our main tool in the construction is the gadget-based trapdoor technique introduced in \cite{MP12}.  Finding a better way to simplify the equality test mechanism would be an interesting future task.

	\end{document}